\documentclass[11pt,oneside,letterpaper]{article}

\usepackage[mode=buildnew]{standalone}
\usepackage{eurosym}
\usepackage{graphicx,amsmath,amsfonts}
\usepackage{setspace}
\usepackage{amssymb}
\usepackage{dsfont}
\usepackage{amsmath}
\usepackage{amsfonts}
\usepackage{ifthen}
\usepackage{mathtools}
\usepackage{graphicx}
\usepackage{enumitem}
\usepackage{color}
\usepackage{framed}
\usepackage[margin=1in]{geometry}
\usepackage{natbib}
\usepackage[dvipsnames]{xcolor}
\usepackage[colorlinks=true, urlcolor=Mahogany, linkcolor=Mahogany, citecolor=Mahogany]{hyperref}
\usepackage{fp}
\usepackage[capitalize]{cleveref}
\usepackage{amsthm}
\usepackage{caption}
\usepackage{subcaption}
\usepackage{palatino}
\usepackage{cuted,balance}
\usepackage{amssymb}
\usepackage{soul}
\usepackage{graphicx,amsmath,amsfonts}
\usepackage{float}
\usepackage{bbm}

\usepackage[ruled]{algorithm2e}

\DeclareMathOperator*{\argmin}{arg\,min}

\newtheorem{assumption}{Assumption}

\newtheorem{definition}{Definition}
\newtheorem{proposition}{Proposition}
\newtheorem{theorem}{Theorem}
\newtheorem{lemma}{Lemma}
\newtheorem{claim}{Claim}

\newtheorem{remark}{Remark}

\allowdisplaybreaks

\let\originaleqref\eqref
\renewcommand{\eqref}{\originaleqref}
\usepackage{bm}

\newcommand{\bx}{\bm{x}}
\newcommand{\by}{\bm{y}}
\newcommand{\ba}{\bm{a}}

\newcommand{\bw}{\bm{w}}
\newcommand{\bxi}{\bm{\xi}}
\newcommand{\bomega}{\bm{\omega}}
\newcommand{\Vx}{V_{\bx}}
\newcommand{\Vy}{V_{\by}}
\newcommand{\Dx}{D_{\bx}}
\newcommand{\Dy}{D_{\by}}
\newcommand{\nux}{\nu_{\bx}}
\newcommand{\nuy}{\nu_{\by}}
\newcommand{\taux}{\tau_{\bx}}
\newcommand{\tauy}{\tau_{\by}}
\newcommand{\htaux}{\hat{\tau}_{\bx}}
\newcommand{\htauy}{\hat{\tau}_{\by}}
\DeclareMathOperator{\sech}{sech}

\usepackage{tikz}
\usetikzlibrary{positioning}
\usepackage{pgfplots}
\usetikzlibrary{intersections, pgfplots.fillbetween}

\begin{document}

\title{A General Framework for Estimating Preferences Using Response Time Data\thanks{We are grateful to Ian Krajbich for comments on the literature and guidance on empirical applications.}}
\author{Federico Echenique\thanks{Department of Economics, University of California, Berkeley} \and Alireza Fallah\thanks{Department of Computer Science, Rice University} \and Michael I. Jordan\thanks{Departments of Electrical Engineering and Computer Sciences and Statistics, University of California, Berkeley; Inria Paris}}

\date{July 2025}

\maketitle

\sloppy

\abstract{We propose a general methodology for recovering preference parameters from data on choices and response times. Our methods  yield estimates with fast ($1/n$ for $n$ data points) convergence rates when specialized to the popular Drift Diffusion Model (DDM), but are broadly applicable to generalizations of the DDM as well as to alternative models of decision making that make use of response time data. The paper develops an empirical application to an experiment on intertemporal choice, showing that the use of response times delivers predictive accuracy and matters for the estimation of economically relevant parameters.

}

\section{Introduction}

Neoclassical economics identifies preferences with choices. That Alice prefers $x$  to $y$ means only that Alice chooses $x$ when presented with the binary-choice problem ``would you like $x$ or $y$?'' As a consequence of this view, and given the great availability of choice data from various sources, economists have developed and used an extensive methodology for analyzing choice data. Choice data are used in estimating preferences, in predicting agents' behavior out of sample, and in conducting counterfactual analysis. There is, however, a growing acceptance in the profession that preferences could be meaningful beyond their interpretation as choices. Many economists now think that non-choice data have a role to play in the estimation of preferences and the prediction of agents' behavior. Now, while there is substantial interest  in the use of non-choice data, there is not yet a mature statistical methodological body of work designed for working with economic non-choice data. Our paper presents progress towards such a methodology.

We develop empirical methodology to estimate preferences from choice and response-time data. Consider Alice again. Now she is presented with a sequence of pairs of alternatives, and makes a choice from each pair. The time she takes to make each decision is recorded. Such data is commonly available in laboratory experiments (a partial list of early examples is \cite{mosteller1951experimental,wilcox1993lottery,rubinstein2007,gabaix2005bounded,gabaix2006costly}), and more recently in data collected at large scale by tech companies (see e.g.\ \cite{xiang2024combining}). Even experiments that were not designed to analyze response time have time stamps that allow for the recovery of response times.  

Our starting point is a very general approach towards estimating parametric preference models using response time and choices. We propose a loss function that is quadratic in a given class of functions of the model parameters, with coefficients given by the response time and the choice variable: this loss function allows us to learn the preference parameters. Importantly, this method does not require distributional assumptions. It only relies on the family of functions being rich enough to include the ratio of expected decision to the expected response time: what we may term the speed-accuracy ratio.

The basic idea behind our estimation strategy is very simple. If the decision is encoded in a scalar $z$, and the response time is $t$, then the quadratic function $tx^2/2 - zx$ is minimized at $x=z/t$. Here $x$ ranges over the values that a function of the data and the parameters can take. In the actual model, the ratio $z/t$ is the expected decision divided by the expected response time: a magnitude that plays an important role in many models of decision making. When the given class of functions of the data is rich enough to include this ratio (Proposition~\ref{proposition:minimizer}), or to approximate it well, we obtain a consistent estimator of the model's preference parameters.

Our general idea is applied to several models of procedural decision making. The most important model in the literature is the \emph{Drift Diffusion Model} (DDM), a model based on a specific sequential sampling procedure (see \cite{Strzalecki_2025} for an exposition directed to an economics audience). The DDM formulates decision making as a process in which evidence for two choices is accumulated over time until a threshold for one of the choices is reached. The speed at which this evidence is gathered is called the drift rate, and once the accumulated evidence hits a decision boundary, a choice is made. For the DDM, the speed-accuracy ratio can be calculated explicitly and used in our general estimation approach. If we specialize to a linear preference model, we obtain a general finite-sample estimator with $1/n$ convergence (Theorem~\ref{theorem:linear_DDM}) in terms of a matrix norm that is suitable for accurate preference predictions (the matrix norm captures how well we estimate the drift rate). In the general non-linear case, we provide a generalization bound for our loss function in terms of Rademacher complexity (Theorem~\ref{theorem:generalization_DDM}). 

The DDM was developed for problems where the strength of a stimulus could be measured objectively (psychometrically or through neurological methods). The problem is then to estimate a scalar parameter that captures how this intensity affects choice. In economic applications, this amounts to estimating utility differences as a stimulus. Our methods, in contrast, seek to estimate the utility as part of our preference parameter exercise. The DDM is particularly interesting for economics because when we ignore response time data, it reduces to the ubiquitous Logit model of discrete choice. This means that our results on the linear DDM offer a perspective on augmenting economic choice data with response times.

Our methods are applicable to models beyond the basic DDM. First, we consider a generalization of the DDM to random starting positions, the so-called extended DDM. In some experimental paradigms, when tasks are difficult and accuracy is low, errors are often observed to be  faster than correct responses. To accommodate such findings, the DDM was generalized to allow for a random initial time \citep{Ratcliff1998}. We show how our methods may be applied to the extended DDM in Section~\ref{sec:E-DDM}. Second, we consider so-called race models. These are quite different from the DDM, and in particular allow for non-binary choices, but our methods are still applicable. Among race models, we focus on a log-normal version of the ``linear ballistic model'' (Section~\ref{sec:LNR}). Here, again, by calculating the speed-accuracy ratio, we are able to apply our general methodology for estimating preferences. 

Finally, we depart from the sequential sampling family of models and in Section~\ref{sec:perceptron} consider a class of threshold discrete choice models in which we may embed some standard models of random utility in discrete choice.  We argue that such models may be estimated through standard techniques for learning linear classifiers, in particular the perceptron algorithm. It is natural in this setting to relate the separation between the two half-spaces implicit in the linear model with response time, but we show that our approach offers advantages in terms of sample complexity and the stringency of the learning objective.

As an illustration,  we develop an empirical application to data on intertemporal choice from \cite{amasino2019amount}; see Section~\ref{sec:experiments}. For each subject in their experiment, we use 100 choices to train a model and then evaluate the performance by predicting 40 choices out of sample. Our results yield remarkably accurate predictions. The average error rate in choice predictions is 6\%, and we can predict response time quite accurately, with 95\% of observed times falling within the predicted mean $\pm$ one standard deviation. We show that the DDM is superior to the log-normal race model we outlined above and that the use of response time leads to better predictive accuracy than a model trained purely on choice data. The application also illustrates that the use of response-time data makes a difference for the substantive empirical results: The use of response time implies larger estimated discount factors. Since this model also predicts behavior better, its parameter estimates are more credible. The DDM, even without the use of response-time data, is superior to the log-normal race model. That we can compare these models is an advantage of our general approach to learning and recovering preferences.

\paragraph{Related Literature.}
The two papers closest to ours are  \cite{fudenberg2020testing} and \cite{sawarni2025}. Both papers propose new statistical techniques for estimating the DDM, and both papers leverage objects that are connected to what we call the accuracy-speed ratio. We differ from these works in that we propose a general methodology that is not exclusively designed for the DDM, but there are several points of close contact between our paper and theirs, which we proceed to discuss.

\cite{fudenberg2020testing} study a DDM with scalar drift, and propose a methodology to test the model and estimate its parameters. They provide necessary and sufficient conditions under which choice probabilities (conditional on stopping at time) and  stopping time distributions are consistent with a version of the DDM. One of their main insights is that the parameter $\delta$ has a ``sample equivalent,'' which they term a ``revealed drift,''  as well as a revealed boundary. The revealed drift is related to the speed-accuracy ratio that we employ in our model, but differs in using relative entropy of the binary choice instead of the expected choice. They are then able to show that a process is consistent with the DDM if and only if the distribution of stopping times coincides with that predicted by the DDM once we plug in these revealed drifts and boundaries. They estimate choice probabilities non-parametrically using a rich collection of functions of time. This is then plugged into their estimator of revealed drift (together with sample analogues for the stopping time). A similar construction is made for the estimator of revealed boundaries. They then propose a $\chi^2$ test for the DDM based on the difference between model moments and sample analogues. 

\cite{sawarni2025} consider the DDM with preference-dependent drift. Like in our model, they are interested in estimating the parameterized drift that results from a parameterization of agents' preferences, and to this end leverage the ratio of expected choice over expected response time.  This framework is similar to ours, but the authors employ a different estimation strategy, with a ``Neyman-orthogonal loss function'' that is distinct from ours. For the linear specification of the drift, which we study in \cref{theorem:linear_DDM} and where we obtain a convergence rate of $\mathcal{O}(1/n)$, they derive asymptotic convergence-in-distribution results for their proposed estimator.
For a broader class of drift functions, they provide non-asymptotic (finite-sample) guarantees by truncating the response time, bounding the excess risk via a critical-radius argument, and controlling the bias introduced by the
truncation of response times. While the objective functions are not directly comparable, the most closely related result in our work is
\cref{theorem:generalization_DDM}, which achieves similar rates using an argument that bounds
the expected maximum response time across the dataset.

A more distantly related contribution is \cite{alos2021time}, who consider a reduced-form model that connects utility differences to response time in binary-choice problems. They provide a condition under which stochastic choices together with response time reveal a preference for one alternative over another. This paper is mainly concerned with ordinal inferences, meaning that one wants to conclude that one alternative is ranked above another, as in traditional revealed preference theory. They do not develop estimation models for parametric models of preferences, which is the focus of our work. That said, we would expect the methods discussed in Section~\ref{sec:perceptron} to be applicable to the model in \cite{alos2021time}.

Turning to other approaches to estimation with response time, \cite{xiang2024combining} develop an application to field data from an online advertising company. They adapt the DDM to their specific setting, where agents are shown an ad before making a decision, and argue that the use of DDM and response time leads to systematically different estimates than what would obtain from a standard Logit model. The estimation is carried out by using simulated non-linear least squares.

\cite{fudenberg2018speed} study how accuracy depends on time for the DDM model and generalizations. Their basic finding is that the expected accuracy conditional on stopping time has the same kind of monotonicity as the boundary. When the boundary is increasing (constant, decreasing) then the conditional accuracy will be increasing (constant, decreasing). More relevant for us, they propose an empirical strategy to estimate the parameters of a DDM using data on the decision value estimates ($W_t$ in our notation). In their application, using data from \cite{Krajbich2010}, they take the final decision values as given by a pre-experiment questionnaire that was administered to gauge the utility of different alternatives (the drift is given by the difference in reported ``ratings'' for the different objects of choice in the KAR experiment). The estimation is then conducted by maximum likelihood.

\cite{echenique2017response} provide a non-statistical revealed-preference test for a model of choice and response time in which response time is used to recover preference intensity. 

At a very basic modeling level, we treat the problem of learning preferences as a classification problem. This approach has some precedents in the literature: 
\cite{chambers2021recovering} propose methods for recovering preferences from choice data based on results in PAC learning. The formulation of preference estimation as a classification problem in machine learning was first proposed by \cite{basu2020}, who establish the learnability of various models of decision making under uncertainty. \cite{chambers2021recovering} build on this formulation to develop results on sample complexity for a Kemeny-distance-based loss function.

We finish our discussion of the literature by giving a brief overview of the role of response time in the economics literature. One use of response time has been to distinguish between deliberate and instinctive responses, in line with the ideas in \cite{kahneman2011thinking}. This literature does not stipulate a computational decision model, but uses response-time data to analyze data collected about specific economic models. \cite{rubinstein2007} advocates for the use of response-time data in economics and argues that it helps distinguish instinctive choices from decisions that require cognitive reasoning. \cite{rubinstein2013} uses response-time data from several standard economic experiments to distinguish deliberate response from mistakes. \cite{agranov2015naive} use constrained time choices (a particular way of capturing response times empirically) to discriminate between naive and sophisticated players in experimental games.
\cite{schotter2021response} argue for the usefulness of response-time data and show it may even be superior to choice data when applied to subsequent prediction of behavior in strategic situations. \cite{CLITHERO201861} provides a survey of the use of response time in economics.

A common procedure in neuroeconomics has been to use a multi-stage experimental design in which, in an initial stage, values or elicited before choices are conducted. A typical example is \cite{milosavljevic2010drift}, who first presents subjects with familiar food items and asks them to rate the items before the actual choice task is started. See also \cite{fehr2011neuroeconomic} for an overview. Our methods obviate the need for such pre-choice tasks, and allow for the estimation of complex preference parameterizations directly from choice and response-time data.

\cite{Krajbich2010} proposed an extension of DDM to account for the endogenous role of attention. The DDM posits that the brain computes a relative decision value that evolves over time, integrating evidence for one item versus the other. The crucial innovation in the attentional DDM is that the drift is not constant but changes dynamically depending on which item is being fixated on at any given moment. \cite{Krajbich2011} extends the attentional DDM beyond binary comparisons, and \cite{Krajbich2012} evaluates the model for economic purchasing decisions.

\section{Problem Formulation}
An agent is offered a choice from a menu $\{\bx,\by\}$ and decides on either $\bx$ or $\by$. To encode the agent's decisions using a variable $z\in\{1,-1\}$, we represent the menu as an ordered pair $(\bx,\by)$. Then we encode the choice of $\bx$ (the ``left'' option) with $z=1$ and of $\by$ (the ``right'' option) with~$-1$.

Suppose that the left alternative is taken from a set $\mathcal{X}\subseteq \mathbb{R}^d$, while the right is from $\mathcal{Y} \subseteq \mathbb{R}^d$ (it is perfectly possible that $\mathcal{X}=\mathcal{Y}$).  We assume these two sets are bounded, and define $D := \sup_{\bx \in \mathcal{X}, \by \in \mathcal{Y}} \|\bx - \by\|$.
For every pair $(\bx, \by) \in \mathcal{X} \times \mathcal{Y}$, let $z(\bx, \by) \in \{1,-1\}$ denote the choice of the agent, which is equal to $1$ if they pick $\bx$ and equal to $-1$ otherwise. Also, let $t(\bx, \by) \in \mathbb{R}_{\geq 0}$ represent the response time of the agent in making their decision between $\bx$ and $\by$. We drop the dependence of functions $z(\cdot, \cdot)$ and $t(\cdot, \cdot)$ on $\bx$ and $\by$ when it is clear by the context.

We assume that pairs of alternatives $(\bx,\by)$ are drawn from a distribution $\mu$ over $\mathcal{X} \times \mathcal{Y}$.
Conditioning on the choices $(\bx,\by)$, and for $z \in \{-1,1\}$, let $p(z;\bx,\by)$ denote the probability of $z(\bx,\by) = z$. Let $F(\cdot;\bx,\by)$ and $f(\cdot;\bx,\by)$ denote the cumulative distribution function (CDF) and the probability density function (PDF) of the distribution of $t(\bx,\by)$, respectively.

For any pair of alternatives $\bx$ and $\by$, we assume that both the distribution of the decision variable and the response time belong to parameterized families of distributions, denoted by $\{p_{\bw}(\cdot;\bx,\by)\}_{\bw}$ and $\{f_{\bw}(\cdot;\bx,\by)\}_{\bw}$, respectively, where $\bw \in \mathcal{W}$ parametrizes each family. We further assume that the true distributions correspond to some universal parameter $\bw^* \in \mathcal{W}$, i.e., $ p(\cdot;\bx,\by)= p_{\bw^*}(\cdot;\bx,\by) $ and $f(\cdot;\bx,\by) = f_{\bw^*}(\cdot;\bx,\by)$.

We assume access to a set of $n$ data points (samples), $\mathcal{S} = {(\bx_i, \by_i, z_i, t_i)}_{i=1}^n$, and our goal is to estimate, or learn, $\bw^*$.

\section{Learning with Response Time} \label{sec:framework}
We design an objective function that enables us to learn the underlying model $\bw^*$ using both the agent’s decisions and response times. Our main idea is to learn the expected accuracy relative to the expected response time. Let us formalize this. For any pair of alternatives $\bx$ and $\by$, we consider a family of functions defined as $\mathcal{G}(\bx,\by) := \{g(\bx,\by;\bw)\}_{\bw}$, parameterized by $\bw \in \mathcal{W}$. We next make the following assumption. 
\begin{assumption}\label{assumption:function_family}
For any pair of alternatives $(\bx,\by) \in \mathcal{X} \times \mathcal{Y}$, we assume that the ratio
\begin{equation} \label{eqn:ratio_choice_time}
\frac{\mathbb{E}[z(\bx,\by)]}{\mathbb{E}[t(\bx,\by)]}
\end{equation}
is finite. Moreover, we assume there exists $\bw^* \in \mathcal{W}$ such that $g(\bx,\by;\bw^*)$ is equal to this ratio.
\end{assumption}
A natural candidate for $\mathcal{G}(\bx,\by)$ is to define $g(\bx,\by;\bw)$ as 
\begin{equation} \label{eqn:ratio_expectations}
\frac{\mathbb{E}_{z \sim p_{\bw}(.;\bx,\by)}[z]}{\mathbb{E}_{t \sim f_{\bw}(.;\bx,\by)}[t]},
\end{equation}
for any $\bw \in \mathcal{W}$, although this is not a required choice; we only need Assumption \ref{assumption:function_family} to hold.
Next, we define the following objective function:\footnote{We assume that this expectation is well-defined. This can be ensured by requiring that $g(\bx,\by;\bw)$ is bounded over $\mathcal{X} \times \mathcal{Y} \times \mathcal{W}$, and that the expected response time is bounded for any pair of alternatives.}
\begin{equation} \label{eqn:main_loss}
\mathcal{L}(\bw):= \mathbb{E}_{\bx,\by,z,t}\left[ \frac{t}{2} g(\bx,\by;\bw)^2 - z g(\bx,\by;\bw) \right ],   
\end{equation}
along with its empirical counterpart:
\begin{equation}\label{eqn:main_loss_empirical}
\hat{\mathcal{L}}(\bw):= \frac{1}{n} \sum_{i=1}^n \left ( \frac{t_i}{2} g(\bx_i,\by_i;\bw)^2 - z_i g(\bx_i,\by_i;\bw) \right ).    
\end{equation}
The following proposition illustrates why \eqref{eqn:main_loss} serves as a meaningful objective for our goal: the underlying model $\bw^*$ is its global minimizer, and, moreover, under mild assumptions, it is the unique one.
\begin{proposition}\label{proposition:minimizer}
Suppose Assumption \eqref{assumption:function_family} holds. Then, $\bw^*$ is the minimizer of $\mathcal{L}(\cdot)$ over $\mathcal{W}$. Moreover, it is the unique minimizer if, for any other $\bw$, functions $g(\bx,\by;\bw)$ and $g(\bx,\by;\bw^*)$ differ with positive probability over $\mathcal{X} \times \mathcal{Y}$.
\end{proposition}
\begin{proof}[Proof of \cref{proposition:minimizer}]
We claim that, for any $\bw \in \mathcal{W}$, we have
\begin{equation} \label{eqn:dif-L}
\mathcal{L}(\bw) - \mathcal{L}(\bw^*) = \frac{1}{2} ~ \mathbb{E}_{(\bx,\by)\sim \mu} \left [
\mathbb{E}[t(\bx,\by)] \left (g(\bx,\by;\bw) - g(\bx,\by;\bw^*) \right)^2
\right].
\end{equation}
Notice that proving this claim would directly show that $\bw^*$ is the minimizer of $\mathcal{L}(\cdot)$. Also, note that given the boundedness of $g(\bx,\by;\bw^*)$, the expected response time is nonzero. Hence, as long as $g(\bx,\by;\bw)$ and $g(\bx,\by;\bw^*)$ differ with positive probability, $\mathcal{L}(\bw) - \mathcal{L}(\bw^*)$ is positive. Therefore, it remains to establish the above claim.

To see this, first, notice that 
\begin{equation} \label{eqn:Loss_double_expectation}
\mathcal{L}(\bw) =  \mathbb{E}_{(\bx,\by)\sim \mu} \left [
\frac{\mathbb{E}[t(\bx,\by)]}{2}g(\bx,\by;\bw)^2 - 
\mathbb{E}[z(\bx,\by)]g(\bx,\by;\bw)
\right].  
\end{equation}
Next, given Assumption \ref{assumption:function_family}, we have
\begin{equation}
\mathbb{E}[z(\bx,\by)] = g(\bx,\by;\bw^*)\mathbb{E}[t(\bx,\by)].     
\end{equation}
Substituting this into \eqref{eqn:Loss_double_expectation}, we obtain
\begin{equation} \label{eqn:Loss_double_expectation_2}
\mathcal{L}(\bw) =  \mathbb{E}_{(\bx,\by)\sim \mu} \left [
\frac{\mathbb{E}[t(\bx,\by)]}{2} 
\left( g(\bx,\by;\bw)^2 - 
2 g(\bx,\by;\bw^*) g(\bx,\by;\bw) \right)
\right].  
\end{equation}
Using this identity, we obtain the above claim. 
\end{proof}
Notice that this particular objective does not require knowledge of the distributions $p(\cdot;\bx,\by)$ or $f(\cdot;\bx,\by)$. It essentially only requires selecting a function class that contains the ratio of expectations specified in \eqref{eqn:ratio_choice_time}.
\begin{remark} 
One may ask what happens if Assumption \ref{assumption:function_family} does not hold, and the family of functions $\mathcal{G}(\bx,\by)$ does not include the ratio $\frac{\mathbb{E}[z(\bx,\by)]}{\mathbb{E}[t(\bx,\by)]}$. To address this, note that, regardless of this condition, we can rewrite $\mathcal{L}(\bw)$ as:
\begin{equation} \label{eqn:recast_loss}
\mathcal{L}(\bw) =  \mathbb{E}_{(\bx,\by)\sim \mu} \left [
\frac{\mathbb{E}[t(\bx,\by)]}{2} 
\left( g(\bx,\by;\bw) - 
\frac{\mathbb{E}[z(\bx,\by)]}{\mathbb{E}[t(\bx,\by)]} \right)^2
\right] -  
\mathbb{E}_{(\bx,\by)\sim \mu} \left [
\frac{\mathbb{E}[z(\bx,\by)]^2}{2\mathbb{E}[t(\bx,\by)]}
\right].
\end{equation}
The second term on the right-hand side does not depend on $\bw$, and therefore, minimizing $\mathcal{L}(\cdot)$ is equivalent to minimizing:
\begin{equation} \label{eqn:recast_loss_2}
\mathbb{E}_{(\bx,\by)\sim \mu} \left [
\frac{\mathbb{E}[t(\bx,\by)]}{2} 
\left( g(\bx,\by;\bw) - 
\frac{\mathbb{E}[z(\bx,\by)]}{\mathbb{E}[t(\bx,\by)]} \right)^2
\right].
\end{equation}
Consequently, minimizing $\mathcal{L}(\bw)$ corresponds to finding a parameter $\bw$ such that, with respect to a weighted average over $\bx$ and $\by$, the function $g(\bx,\by;\bw)$ closely approximates the ratio $\frac{\mathbb{E}[z(\bx,\by)]}{\mathbb{E}[t(\bx,\by)]}$.
\end{remark}
In general, the optimization problem in \eqref{eqn:main_loss} may be non-convex. However, in what follows, we discuss how it can still be useful for modeling popular decision-making processes such as drift-diffusion models and race models.
\section{Drift-Diffusion Models} \label{sec:DDM}

The most important procedural model of decision making in the quantitative decision sciences is, arguably, the Drift Diffusion Model (DDM) that we discussed in the introduction. We show that our general methodology yields a method for recovering the parameters of the DDM. 

The DDM assumes a drifted Brownian motion and two boundaries, each representing one of the two choices, and whichever boundary the diffusion process hits first determines the agent’s choice. More formally, consider the following diffusion process:
\begin{equation} \label{eqn:DDM_model_original}
W_t = B_t + v(\bx,\by;\bw^*) t,
\end{equation}
where $B_t$ is a standard Brownian motion (starting at 0) and $v(\bx,\by;\bw)$ is the drift associated with the two alternatives $\bx$ and $\by$, parametrized by $\bw$. Suppose there are two constant boundaries at $b$ and $-b$ for some positive $b$. The agent chooses between the two options depending on which of the two boundaries is hit first. More specifically, let $T^b_+$ denote the first time that the process hits the boundary $b$, i.e.,
\begin{equation} \label{eqn:T_plus_b}
T^b_+ := \inf\{t : W_t = b\},
\end{equation}
with the convention that this could be infinite if the process never hits $b$. Similarly, define $T^b_{-}$ as the first time the process hits the boundary $-b$. The agent chooses the first (second) option $\bx$ ($\by$) if $T^b_+ < T^b_-$ ($T^b_- < T^b_+$). We then define the hitting time
\begin{equation} \label{eqn:T_b}
T^b := \min\{T^b_+, T^b_-\}.
\end{equation}
One can verify that $T^b$ is almost surely finite, meaning that the agent almost surely makes a decision.

Lemma~\ref{lemma:BM_distributions}, stated in Section~\ref{sec:proofs} and proven in the appendix, calculates the distribution of $z(x,y)$ and $t(x,y)$ under the DDM model. Interestingly, the distribution of decision time does not depend on whether the first or second choice is made: under the DDM model, and for a fixed drift and boundary, the distribution of the time it takes to reach a decision is independent of the decision itself. As a consequence, the ratio of expected choice over expected response time, \eqref{eqn:ratio_choice_time} equals the drift rate over the boundary, 
$$\frac{v(\bx,\by;\bw^*)}{b}.$$ 

A natural candidate for the class of functions $\mathcal{G}(\bx,\by)$ in Section~\ref{sec:framework} is therefore to take  $g(\bx,\by;\bw) = \frac{v(\bx,\by;\bw)}{b}$. In this case, the empirical loss $\hat{\mathcal{L}}(\bw)$ is given by
\begin{equation} \label{eqn:empirical_L_DDM}
\hat{\mathcal{L}}(\bw)= \frac{1}{n} \sum_{i=1}^n \left( \frac{t_i}{2} \left( \frac{v(\bx_i,\by_i;\bw)}{b}\right)^2- z_i \frac{v(\bx_i,\by_i;\bw)}{b} \right).    
\end{equation}

The joint distribution of $(z,t)$ is (\cref{lemma:BM_distributions}) given by
\begin{equation}
\frac{1}{2} \exp \left (bzv(\bx,\by;\bw^*) - \frac{v(\bx,\by;\bw^*)^2}{2}t  \right) \varphi(t),    
\end{equation}
and therefore, minimizing $\hat{\mathcal{L}}(\bw)$ in \eqref{eqn:empirical_L_DDM} coincides with finding the maximum likelihood estimator (MLE) of $\bw$, up to a $b^2$ factor (this does not hold in generalizations of the DDM, as we will see shortly in \cref{sec:E-DDM}.)

We discuss quantitative bounds on the approximation of the true parameters in the model for the case when the drift is linear in object attributes. Thus, \begin{equation} \label{eqn:linear_drift}
v(\bx,\by;\bw^*) = (\bx - \by)^\top \bw^*.
\end{equation}
Importantly, the approximation guarantee that we establish below is in terms of a particular matrix norm that is meant to capture the drift component of the DDM. In other words, when we evaluate how far the estimated value of the vector $\bw$ is from the true underlying $\bw$, we use a metric that takes into account that we mean to use these parameters to approximate the drift of the DDM.

The loss function $\hat{\mathcal{L}}(\bw)$ is here a quadratic function of $\bw$. The following result shows how a stochastic gradient descent method can lead to a good approximation of $\bw^*$. The proof can be found in \cref{sec:proofsddm}.
\begin{theorem} \label{theorem:linear_DDM}
Consider the empirical loss $\eqref{eqn:empirical_L_DDM}$ with the linear drift $\eqref{eqn:linear_drift}$. We run the gradient descent algorithm initialized with an arbitrary $\hat{\bw}_0 \in \mathcal{W}$ and following the update rule:
\begin{equation}
\hat{\bw}_{i}  = \hat{\bw}_{i-1} - \frac{\lambda}{b} \left ( 
t_i \hat{\bw}_{i-1}^\top (\bx_i-\by_i) - z_i \right) (\bx_i-\by_i),
\end{equation}
where $\lambda = 1/(8D^2)$. Then, we have the bound:
\begin{equation}
\mathbb{E} \left[ \left \|\frac{1}{n+1}\sum_{i=0}^{n}\hat{\bw}_i - \bw^* \right \|_{\Sigma}^2 \right] \leq \frac{8}{n+1} \sqrt{1+b^2 D^2\|\bw^*\|^2} \left (\frac{d}{b^2} + 2D^2 \|\hat{\bw}_0 - \bw^*\|^2 \right),   
\end{equation}
where the expectation is taken over the randomness in the draw of the samples, and $\|\cdot\|_\Sigma$ denotes the matrix norm with respect to the matrix $\Sigma := \mathbb{E}_\mu[(\bx - \by)(\bx - \by)^\top]$.
\end{theorem}
One might wonder whether off-the-shelf optimization results could be used to derive a straightforward bound on the error of estimating $\bw^*$, given that the objective function is quadratic. It is worth noting, however, that classical optimization results for strongly convex objective functions typically provide a bound of the form $\mathcal{O}(d\kappa/n)$, where $\kappa$ here is proportional to the inverse of the minimum eigenvalue of $\Sigma$. If the minimum eigenvalue is zero, this bound deteriorates to $\mathcal{O}(d/\sqrt{n})$. In contrast, the result we present here does not depend on the minimum eigenvalue of $\Sigma$.

It is also worth highlighting that our result is expressed in terms of the matrix norm with respect to $\Sigma$. We argue that, in this linear DDM model, this is the most relevant metric for evaluating how accurately the underlying model $\bw^*$ is learned, since $\bw^*$ influences the model through the drift rate $(\bx - \by)^\top\bw^*$. Therefore, $\|\hat{\bw} - \bw^*\|_\Sigma^2$ effectively measures the expected error in estimating the drift rate. In \cref{sec:perceptron}, we further discuss how this result is instrumental in deriving high-probability error bounds for predicting the agent’s choices. 

\paragraph{No response-time data}
If we exclusively use choice data, ignoring response times, then the methods we have outlined for the linear DDM reduce to the standard methodology in economics of estimating a Logit discrete choice model \citep{mcfadden1974}. 

\begin{remark} \label{remark:logistic}
The maximum likelihood estimator of the DDM model under a linear drift yields (see \cref{lemma:BM_distributions} of Section~\ref{sec:proofs})  the following logistic regression problem:
\begin{equation} \label{eqn:MLE_with_z}
\argmin_{\bw \in \mathcal{W}} \frac{1}{n} \sum_{i=1}^n \log \left(1 + \exp\left(-2b \bw^\top (\bx_i - \by_i) z_i \right) \right).
\end{equation}
Standard results in the optimization literature allow us to estimate $\bw^*$ at a rate of $\mathcal{O}(d\kappa/n)$, where $\kappa$ is proportional to the inverse of the strong convexity parameter of the objective function at the optimum $\bw^*$ \citep{bach2014adaptivity}. This dependence can be further relaxed under additional assumptions \citep{bach2013non}.
\end{remark}

The logistic regression is useful in an identification strategy when we seek to predict response times, not just choices. \cref{theorem:linear_DDM} allows us to effectively learn $\bw^*/b$ rather than $\bw^*$, which is sufficient to predict choices in the DDM. If we also want to predict response times, then the logistic loss in \eqref{eqn:MLE_with_z} serves to learn $b\bw^*$. By combining these two approaches, we can recover estimates for both $\bw^*$ and $b$.

\subsection{Beyond the Linear Drift \label{sec:Rademacher_DDM}}
In the previous section, we provided theoretical guarantees for the case where the drift function $v(\bx, \by; \bw)$ admits a linear form, as in Equation~\eqref{eqn:linear_drift}. When the drift is not linear, the optimization problem \eqref{eqn:empirical_L_DDM} may become non-convex. Nonetheless, gradient-based methods can still be employed, and there exist guarantees for finding local minima \citep{ge2015escaping, lee2016gradient}, and even global minima under certain structural conditions \citep{kawaguchi2016deep, du2019gradient}.

However, having obtained a solution to the empirical loss \eqref{eqn:empirical_L_DDM}, it remains to study how well this solution will generalize to the out-of-sample loss $\mathcal{L}(\cdot)$. More formally, suppose an optimization algorithm returns $\hat{\bw}$ as a solution to \eqref{eqn:empirical_L_DDM}. What can we say about $\mathbb{E}[\mathcal{L}(\hat{\bw}) - \mathcal{L}(\bw^*)]$? In this section, we address this question in the context of the DDM model. 

Our analysis relies on the Rademacher complexity~\citep{Bartlett}. To define it, let $\bomega = (\omega_i)_{i=1}^n$ denote $n$ independent random variables with the Rademacher distribution; that is, each $\omega_i$ is $1$ with probability $1/2$ and $-1$ with probability $1/2$. The Rademacher complexity of the class of functions $\{v(\cdot, \cdot; \bw)\}_{\bw}$ with respect to a dataset $\tilde{\mathcal{S}} := {(\tilde{\bx}_i, \tilde{\by}_i)}_{i=1}^n$ is defined as:
\begin{equation} \label{eqn:Rademacher_def}
\mathcal{R}_n \left (\{v(\cdot, \cdot; \bw)\}_{\bw} \circ \tilde{\mathcal{S}} \right):= 
\frac{1}{n} \mathbb{E}_{\bomega}\left[ \sup_{\bw \in \mathcal{W}} \sum_{i=1}^n \omega_i ~v(\tilde{\bx}_i,\tilde{\by}_i;\bw) \right],
\end{equation}
where the expectation is taken over the distribution $\bomega$. The Rademacher complexity has been computed for many popular classes of functions. For instance, if we generalize the linear drift model considered earlier and take
\begin{equation}\label{eqn:extended-linear-class}
v(\bx, \by; \bw) = \Psi(\bw^\top [\bx^\top, \by^\top]),    
\end{equation}
where $\bw$ is now a $2d$-dimensional vector and $\Psi(\cdot)$ is an $L$-Lipschitz function, then the Rademacher complexity defined in \eqref{eqn:Rademacher_def} is upper bounded by
\begin{equation}
\frac{L \|\mathcal{W}\|_2 \sqrt{\|\mathcal{X}\|_2^2+\|\mathcal{Y}\|_2^2}}{\sqrt{n}},
\end{equation} 
where $\|\mathcal{W}\|_2 := \sup_{\bw \in \mathcal{W}} \|\bw\|_2$, and $\|\mathcal{X}\|_2$ and $\|\mathcal{Y}\|_2$ are defined analogously~\citep[see, for example,][Chapter 26]{shalev2014understanding}.\footnote{This framework can be further extended to more complex function classes such as neural networks; see, for example, \citet{golowich2018size}.} 
Next, we state our main result, which connects the value of $\mathcal{L}(\cdot)$ at the approximate solution of the empirical objective \eqref{eqn:empirical_L_DDM} to the Rademacher complexity of the class of drift functions defined above.
\begin{theorem}\label{theorem:generalization_DDM}
Consider the DDM model in \eqref{eqn:DDM_model_original}, and suppose that
\begin{equation}
\sup_{\bx \in \mathcal{X}, \by \in \mathcal{Y}, \bw \in \mathcal{W}} |v(\bx,\by;\bw)| \leq K.    
\end{equation}
Then, we have the generalization bound:
\begin{equation} \label{eqn:DDM_generalization_bound}
\mathbb{E}_{\mathcal{S}} \left [ \sup_{\bw \in \mathcal{W}} \left(\mathcal{L}(\bw) - \hat{\mathcal{L}}(\bw) \right) \right]  
\leq 2\left( \frac{1}{b}+ CK\log(n) \right) \mathbb{E} \bigg [\mathcal{R}_n \Big (\{v(\cdot, \cdot; \bw)\}_{\bw} \circ (\bx_i,\by_i)_{i=1}^n \Big) \bigg] ,
\end{equation}
for some constant $C$ (independent of the model parameters), where the expectation is taken over the randomness in drawing the sample $\mathcal{S} = {(\bx_i, \by_i, z_i, t_i)}_{i=1}^n$. In particular, this implies that if $\hat{\bw}_{\mathcal{S}}$ is an approximate solution to the empirical objective \eqref{eqn:empirical_L_DDM}, with expected error $\text{Err}$, then
\begin{equation} \label{eqn:DDM_generalization_bound_2}
\mathbb{E}_{\mathcal{S}} \left [ \mathcal{L}(\hat{\bw}_{\mathcal{S}}) - \mathcal{L}(\bw^*) \right]  \leq \text{Err } + 2\left( \frac{1}{b}+ CK\log(n) \right) \mathbb{E} \bigg [\mathcal{R}_n \Big (\{v(\cdot, \cdot; \bw)\}_{\bw} \circ (\bx_i,\by_i)_{i=1}^n \Big) \bigg].   
\end{equation}
\end{theorem}
See \cref{sec:proofsddm} for the proof.
Note that \cref{theorem:generalization_DDM} implies that any bound on the Rademacher complexity of the class of drift functions yields---up to a $\log(n)$ factor---a bound on the generalization error of the empirical solution of $\hat{\mathcal{L}}(\cdot)$ with respect to the out-of-sample objective $\mathcal{L}(\cdot)$. For example, if we consider the class of drift functions in the form of \eqref{eqn:extended-linear-class}, we obtain a bound of order $\log(n)/\sqrt{n}$.

Finally, it is worth noting that \cref{theorem:generalization_DDM} offers a bound on the expectation of $\mathcal{L}(\hat{\bw}_{\mathcal{S}}) - \mathcal{L}(\bw^*)$, though our interest might lie more directly in the difference between $\hat{\bw}_{\mathcal{S}}$ and $\bw^*$. Since $\bw^*$ influences the DDM model exclusively via $v(\bx,\by;\bw^*)$, it becomes essential for learning $\bw^*$ to assume that $v(\cdot, \cdot;\bw^*)$ is distinguishable from $v(\cdot, \cdot;\bw)$ for other values of $\bw$. This is formalized in the following lemma whose proof is deferred to the appendix:
\begin{lemma}\label{lemma:generalization-DDM}
Suppose 
\begin{equation}
\sup_{\bx \in \mathcal{X}, \by \in \mathcal{Y}, \bw \in \mathcal{W}} |v(\bx,\by;\bw)| \leq K, 
\quad \text{and }~~
\mathbb{E}_{(\bx,\by)\sim\mu} \left [ \Big( v(\bx,\by;\bw) - v(\bx,\by;\bw^*) \Big)^2 \right] \geq \|\bw-\bw^*\|_{*}^2,
\end{equation}
for some arbitrary norm $\|\cdot\|_{*}$ and all $\bw \in \mathcal{W}$.
Then, for any $\bw \in \mathcal{W}$, we have  
\begin{equation}
\| \bw - \bw^* \|_{*}^2 \leq 2\sqrt{1+b^2K^2} ~\left( \mathcal{L}(\bw) - \mathcal{L}(\bw^*) \right).
\end{equation}  
\end{lemma}
\subsection{Extended Drift-Diffusion Model}\label{sec:E-DDM}

Subsequent to the success of the DDM in explaining human behavior in binary-choice tasks, several extensions have been proposed to allow for trial-to-trial variability in its parameters \citep{ratcliff2002estimating}. One of the most common relaxations is to assume that the initial starting point, rather than being fixed at zero, is drawn from a random distribution. The notion of a random starting time to accommodate the difficulties in sequential sampling models was proposed by \cite{Laming1968} (in the context of the discrete-time precursors to the DDM). For the DDM, \cite{Ratcliff1998}  explicitly incorporated a parameter for across-trial variability in the starting point. This integration was a critical step in the development of what is now often called the ``full DDM.''

The extended DDM model in \eqref{eqn:DDM_model_original} is now modified to: 
\begin{equation} \label{eqn:E-DDM_model}
W_t = B_t + v(\bx,\by;\bw^*) t + \zeta,
\end{equation}
where $\zeta$ represents the initial bias, drawn from a mean-zero distribution $\pi_0(\cdot)$ with support contained in $(-b, b)$, and $B_t$ is standard Brownian motion (starting from zero). 

We next try to answer how can we learn the underlying parameter $\bw^*$. As a first step, let us examine what the MLE estimator (which corresponds to minimizing $\mathcal{L}(\cdot)$ in the original DDM model) yields. With a slight abuse of notation, we reuse the definitions of $T_{+}^b$ and $T^b$ from \eqref{eqn:T_plus_b} and \eqref{eqn:T_b}, this time in the context of the extended DDM model. The proof of the following lemma is provided in the appendix.
\begin{lemma} \label{lemma:E-DDM-joint-pdf}
Under the extended DDM model \eqref{eqn:E-DDM_model}, the joint distribution of $T^b=t$ and $z=1$ is given by
\begin{equation}
\frac{\exp(bv - {v^2}/{2}t)}{\sqrt{2\pi t^3}} \int_{-b}^b \pi_0(\zeta) \sum_{n=-\infty}^\infty \left[ (2nb+b-\zeta) \exp \left (-v \zeta - \frac{(2nb+b-\zeta)^2}{2t} \right ) \right]  d\zeta,
\end{equation}
with $v$ denoting $v(\bx,\by;\bw^*)$ to simplify the expression.
\end{lemma}
We can derive the joint distribution for $z = -1$ in a similar manner. As one can see, the interdependence between $v$ and $\zeta$ indicates that (i) the MLE estimator for $\bw$ would depend on the distribution $\pi_0(\cdot)$, and (ii) even for simple choices of $\pi_0(\cdot)$, such as the uniform distribution, the estimation remains intractable. 

Let us now examine what our framework yields. The following proposition characterizes the ratio of expected choice to expected response time, as defined in \eqref{eqn:ratio_choice_time}, in this setting. The proof is presented in \cref{sec:proofsddm}.
\begin{proposition} \label{proposition:e-DDM-ratio}
Under the extended DDM model \eqref{eqn:E-DDM_model}, and for every pair of alternatives $(\bx,\by)$, we have
\begin{equation}
\frac{\mathbb{E}[z(\bx,\by)]}{\mathbb{E}[t(\bx,\by)]} = \frac{v(\bx,\by;\bw^*)}{b}.   
\end{equation}
\end{proposition}
\cref{proposition:e-DDM-ratio} implies that even for the extended DDM, the loss function $\mathcal{L}(\cdot)$ retains the same form as in the original DDM, given in \eqref{eqn:empirical_L_DDM}. In particular, in the case of linear drift, this reduces to a quadratic optimization problem. This is a remarkable result: while we previously observed that the MLE estimator could be intractable, here we do not even require any knowledge of the initialization distribution $\pi_0(\cdot)$ to estimate $\bw^*$.

\section{Race Models} \label{sec:LNR}
Race models are a type of sequential sampling model that propose that decisions are made by accumulating evidence for different alternatives until one reaches a threshold. In particular, linear deterministic accumulator models, such as the Linear Ballistic Accumulator (LBA) \citep{brown2008simplest} and the Lognormal Race model (LNR) \citep{heathcote2012linear}, are a popular class of race models. In these models, each alternative is associated with a linear trajectory that has a random starting point and a random drift. The alternative whose trajectory first reaches the decision threshold is selected. In this section, we focus primarily on the LNR model, which we now describe formally.

For alternative $\bx$, a corresponding linear process starts at a random initial point, located at distance $\Dx$ from the boundary, and proceeds with drift (slope) $\Vx$. This represents the evidence accumulation over time. Consequently, this process hits the boundary at time $\taux := \Dx/\Vx$. The parameters $\Dy$, $\Vy$, and $\tauy$ are defined similarly for the other alternative. In this model, the agent selects alternative $\bx$ (i.e., $z(\bx,\by)=1$) if $\taux \leq \tauy$, and selects alternative $\by$ otherwise. The response time $t(\bx,\by)$ is then given by $\min\{\taux,\tauy\}$.\footnote{Although this paper focuses on the two-alternative case, once can see that this model, unlike the DDM, is applicable to multiple-choice selections as well. Our proposed framework in \cref{sec:framework} also extends to that setting.}

The LNR model assumes that the joint distribution of $\Dx$ and $\Dy$, as well as the joint distribution of $\Vx$ and $\Vy$, follow a log-normal distribution. In particular, we assume that $\Dx$ and $\Dy$ are independent and both drawn from $\text{Lognormal}(D_0, 1)$, and that the drift rates are given by
\begin{equation}\label{eqn:log-normal-dist}
(\Vx, \Vy) \sim \text{Lognormal} \left( 
\begin{bmatrix} \nu(\bx;\bw^*) \\ \nu(\by;\bw^*) \end{bmatrix},
\begin{bmatrix}
1 & \rho \\
\rho & 1
\end{bmatrix}
\right),
\end{equation}
where $\rho$ denotes the correlation between the two races’ drift rates and is one of the key features of the LNR model that interconnects the two races.
It is worth noting that our results in this section extend to settings in which the drift rate variances differ or the initial distances are correlated. However, we focus on the above case to keep the expressions simpler and the insights easier to follow.

The next result characterizes the ratio of the expected choice probability to the expected response time in the LNR model.
\begin{proposition} \label{proposition:LNR-ratio}
Consider the LNR model described above. Then, for any two alternatives \((\bx, \by)\), the expected choice probability and expected response time are given by:
\begin{subequations}\label{eqn:expected-z-t-LNR}
\begin{align} 
\mathbb{E}[z(\bx,\by)]& = 2 \Phi \left( \frac{\nu(\bx;\bw^*) - \nu(\by;\bw^*)}{\sqrt{4-2\rho}} \right)-1,  \label{eqn:expected-z-LNR} \\
\mathbb{E}[t(\bx,\by)]& = \exp \left(D_0 +1 -\frac{\nu(\bx;\bw^*) + \nu(\by;\bw^*)}{2} \right) J_{\rho}\left( \nu(\bx;\bw^*) - \nu(\by;\bw^*)\right), \label{eqn:expected-t-LNR}
\end{align}
\end{subequations}
where $\Phi(\cdot)$ is the CDF of the standard normal distribution, and $J_{\rho}(\cdot)$ is given by
\begin{equation}
J_{\rho}(r):= \exp(r/2) \Phi \left(\frac{-r+\rho-2}{\sqrt{4-2\rho}} \right) + \exp(-r/2)\Phi \left( \frac{r+\rho-2}{\sqrt{4-2\rho}} \right). 
\end{equation}
\end{proposition}
See \cref{sec:proofsrace} for the proof.
This result allows us to define a parameterized class of functions for the ratio of the expected decision to the expected response time in LNR models. Although the resulting optimization problem is non-convex, it can still be solved numerically, as demonstrated by an example in \cref{sec:experiments}.

\section{Learning Halfspaces} \label{sec:perceptron}
We now turn our attention to a class of threshold models. Consider the case where the agent’s decision follows a linear model, up to some noise. More specifically, if an agent is presented with two alternatives $\bx$ and $\by$, they choose between them according to the sign of $(\bx - \by)^\top \bw^*$, but may deviate or make a mistake with some probability $\eta(\bx, \by)$. More formally, we have
\begin{equation}
p(z;\bx, \by) =
\begin{cases}
1 - \eta(\bx, \by) & \text{if } z = \text{sign}((\bx - \by)^\top \bw^*) \\
\eta(\bx, \by) & \text{otherwise},
\end{cases}
\end{equation}
for some \textit{unknown} error function $\eta : \mathcal{X} \times \mathcal{Y} \to [0, \bar\eta]$, where the \textit{known} parameter $\bar\eta < \frac{1}{2}$ provides an upper bound on the probability of making a mistake. This is known as the problem of learning half-spaces with Massart noise \citep{massart2006risk}. A special case of this model in which $\eta(\bx,\by)=\bar\eta$, i.e., the probability of mistake is constant and regardless of alternatives, is known as the Random Classification Noise (RCN) setting. 

An advantage of this basic setting over the DDM and Race models is that it does not require the assumption of a procedural model of decision making. Suppose that we operate under this model and do not observe the response time. We are interested in understanding how accurately we can estimate the agent’s decisions on new, unseen data. This problem has been studied extensively in the learning theory literature, beginning with the Perceptron algorithm in the noiseless setting \citep{rosenblatt1958perceptron}, and continuing with the work of \cite{bylander1994learning} and \cite{blum1998polynomial} on the design of polynomial-time algorithms in the RCN setting.\footnote{The literature on learning halfspaces typically focuses on labeling a single vector $\bx$, rather than the difference between two alternatives $\bx - \by$, as in our setting. However, for convenience, we translate their results into our framework when presenting them.
} Recent work has focused on determining how many samples are needed to achieve guarantees of the following form:
\begin{definition} \label{definition:high_prob_predictiton}
For any $\alpha, \delta \in (0,1]$, we call an algorithm an ($\alpha$, $\delta$)-predictor if, with probability at least $1 - \delta$, it returns a function $h:\mathcal{X} \times \mathcal{Y} \to \{-1, +1\}$ such that the probability of predicting differently from the agent’s decision is at most $\alpha$, i.e.,
\begin{equation}
\mathbb{P} \left(h(\bx,\by) \neq  z(\bx,\by) \right) \leq \alpha,
\end{equation}
where the probability is over both the realization of alternatives $(\bx, \by)$ as well as the randomness in the agent’s decision $z(\cdot,\cdot)$. 
\end{definition}
One could naturally consider an alternative criterion that focuses on how well we can estimate the underlying model, rather than the agent's action. The following definition captures this idea.
\begin{definition} \label{definition:high_prob_accuracy}
For any $\varepsilon, \delta \in (0,1]$, we call an algorithm an ($\varepsilon$, $\delta$)-estimator if, with probability at least $1 - \delta$, it returns a function $h:\mathcal{X} \times \mathcal{Y} \to \{-1, +1\}$ such that the probability of estimating an output that differs from the true model $\bw^*$ is at most $\varepsilon$, i.e.,
\begin{equation}
\mu \left( \left\{ (\bx,\by) : h(\bx,\by) \neq  \text{sign}((\bx - \by)^\top \bw^*) \right \} \right) \leq \varepsilon.
\end{equation}
\end{definition}
In the RCN setting, one can verify that these two definitions are, in fact, equivalent, in the sense that an algorithm is an ($\varepsilon$, $\delta$)-estimator, if and only if, it is an ($(1-2\eta) \varepsilon + \eta$, $\delta$)-predictor. However, under Massart noise, the relationship between the two notions is less straightforward. Most of the work in the literature addressing Massart noise proposes algorithms that are $(\eta + \varepsilon, \delta)$-predictors \citep{diakonikolas2019distribution, diakonikolas2024near, chandrasekaran2024learning}. This can, of course, result from an algorithm that is a $(\frac{\varepsilon}{1 - 2\eta}, \delta)$-estimator. However, in the other direction, the latter may actually provide a stronger guarantee in terms of prediction---especially if the upper bound $\eta(\bx, \by) \leq \eta$ is loose. Finally, it is worth emphasizing that neither of these two definitions necessarily implies learning the underlying parameter $w^*$.

Let us now briefly discuss the state-of-the-art results under both RCN and Massart noise.  A common assumption in many papers studying the halfspace learning problem is that the data satisfies a margin property---meaning that, with probability one, we have $|(\bx - \by)^\top \bw^*| \geq \gamma$ for some margin $\gamma > 0$. Under this assumption, the recent work of \cite{diakonikolas2024near, chandrasekaran2024learning} considers a setting with Massart noise and $D=1$. They propose an $(\eta + \varepsilon, 0.1)$-predictor that requires $\tilde{\mathcal{O}}(1 / (\gamma^2 \varepsilon^2))$ samples. This improves upon the result of \cite{diakonikolas2019distribution}, which requires $\tilde{\mathcal{O}}(1 / (\gamma^3 \varepsilon^5))$ samples, and matches the sample complexity achieved by \cite{diakonikolas2023information} under RCN (again for $(\eta + \varepsilon, 0.1)$-predictors). \footnote{It is well known that any result with a finite probability of failure (0.1 in this case) can be transformed into a result with failure probability $\delta$, by increasing the sample complexity by a factor of $\mathcal{O}(\log(1/\delta))$. We further explain how this works in the proof of \cref{proposition:translate-to-high-probability}.}
Without the margin assumption, but under the additional assumption that each coordinate of $\bx - \by$ is a $b$-bit integer (i.e., $\bx - \by \in \{1, 2, \cdots, 2^b\}^d$), and assuming Massart noise with $D = 1$, \cite{diakonikolas2019distribution} propose an $(\eta + \varepsilon, 0.1)$-predictor that requires $\tilde{\mathcal{O}}(d^3 b^3 / \varepsilon^5)$ samples.

As \cref{lemma:BM_distributions} shows, the case of the DDM model with linear drift in \cref{sec:DDM} can be seen as a problem of learning halfspaces with Massart noise, with
\begin{equation}
\eta(\bx,\by) = \frac{1}{1+\exp \left(2b|(\bx-\by)^\top \bw^*|\right)}.
\end{equation}

It is worth noting that our result in \cref{theorem:linear_DDM} provides a stronger guarantee than \cref{definition:high_prob_predictiton} or \cref{definition:high_prob_accuracy}, since it implies convergence to the true model (which does not necessarily follow from those definitions). That said, we next discuss how our result translates to \cref{definition:high_prob_accuracy} (which, as discussed earlier, can in turn be translated to guarantees in the form of \cref{definition:high_prob_predictiton}). The proof can be found in \cref{sec:proofhighprob}.
\begin{proposition} \label{proposition:translate-to-high-probability}
Consider the DDM model in \cref{sec:DDM} with linear drift \eqref{eqn:linear_drift}, and assume that $D \leq 1$ and $\|\bw^*\| \leq 1$.\footnote{This result easily extends to any bound on $D$ and $\|\bw^*\|$; however, we make this assumption to facilitate comparison with the results above.} Then, for any positive $\varepsilon, \delta$, and $\gamma$,
\begin{equation}
\mathcal{O}\left(\frac{\log(1/\delta)d}{\varepsilon \gamma^2 \min\{b,b^2\}}\right)    
\end{equation}
samples are sufficient to find an $(\tilde{\varepsilon}, \delta)$-estimator, where
\begin{equation}
\tilde{\varepsilon} := \varepsilon + \mu \left( \left\{(\bx,\by): |(\bx-\by)^\top \bw^*| < \gamma \right\}\right).    
\end{equation}
\end{proposition}

Let us compare this result with the existing results on learning halfspaces which we stated earlier. First, our result focuses on recovering the underlying parameter $\bw^*$, which, as we discussed earlier, is a stronger guarantee than merely predicting the agent’s choices correctly. Moreover, we do not require any large-margin assumption, since the result of \cref{proposition:translate-to-high-probability} holds for any $\gamma$. For example, depending on the underlying distribution, one can optimize over $\gamma$ or choose it sufficiently small to achieve an overall $2\varepsilon$ error.

For the sake of comparison, if the data satisfies a $\gamma$-margin condition (meaning that $|(\bx - \by)^\top \bw^*| \geq \gamma$ almost surely for some $\gamma$), then our sample complexity in terms of its dependence on $\varepsilon$ improves the existing $1/\varepsilon^2$ rate to $1/\varepsilon$ under the assumed DDM model. It is worth noting that, in the standard learning halfspaces problem, and when some margin condition holds, the dependence on the dimension is typically removed using random projections, such as those given by the Johnson–Lindenstrauss lemma~\citep[see, e.g.,][]{diakonikolas2023information}. However, since our goal is to recover the true $\bw^*$, we do not employ such dimensionality reductions in this paper.
\section{Empirical Application} \label{sec:experiments}

As an illustration of our methodology, we develop an empirical application using data on intertemporal choice collected by \cite{amasino2019amount}. Specifically, we use their replication dataset, which includes data from one hundred human subjects (agents). For each agent, there are approximately 141 choice observations. Each observation presents two alternatives to the individual: one option offers an immediate monetary amount $m_i$, while the other option offers ten dollars after a delay of $t_d$ units of time. Consequently, each pair of choices can be represented as two-dimensional vectors: $[m_i, 0]$ and $[10, t_d]$, where the first entry denotes the monetary amount and the second entry denotes the time delay.\footnote{This is the paradigm of \textit{dated rewards} introduced by \cite{lancaster63} and \cite{fishrubin89}.}

For each of the 100 agents, we use the first 100 data points to train the model and the remaining data points (approximately 40) for testing, as described below. It is important to note that we train a separate model from scratch for each agent. More specifically, we train three models for each agent. The first two are based on the DDM framework: one uses only the choice data, while the other uses both choice and time data (corresponding to our formulation in \eqref{eqn:main_loss_empirical}). The third model corresponds to the LNR model with $D_0$ and $\rho$ set to zero.
The model that only uses choice data corresponds to the standard procedure in economics (it reduces to the logistic loss discussed in Remark~\ref{remark:logistic}).

For each agent, we measure how well the trained models can predict the subject’s choices and their response times, as detailed below.

\paragraph{Predicting the agent's choice:}
To evaluate the models’ performance in predicting the agents' choices, we compute the fraction of the approximately 40 test samples for which the model predicts the agent’s choice incorrectly. We refer to this fraction as the error rate. The empirical CDF of the error rate for all three models is shown in \cref{fig:prediction-error-CDF}, and the histogram for the two DDM models is also shown in \cref{fig:prediction-error-Hist}. The average error rates for the DDM model using both choice and response-time data, the DDM model using only choice data, and the LNR model are given by 0.0627, 0.0754, and 0.1215, respectively.

\begin{figure}
\centering
\begin{subfigure}{.49\textwidth}
  \centering
  \includegraphics[width=1\linewidth]{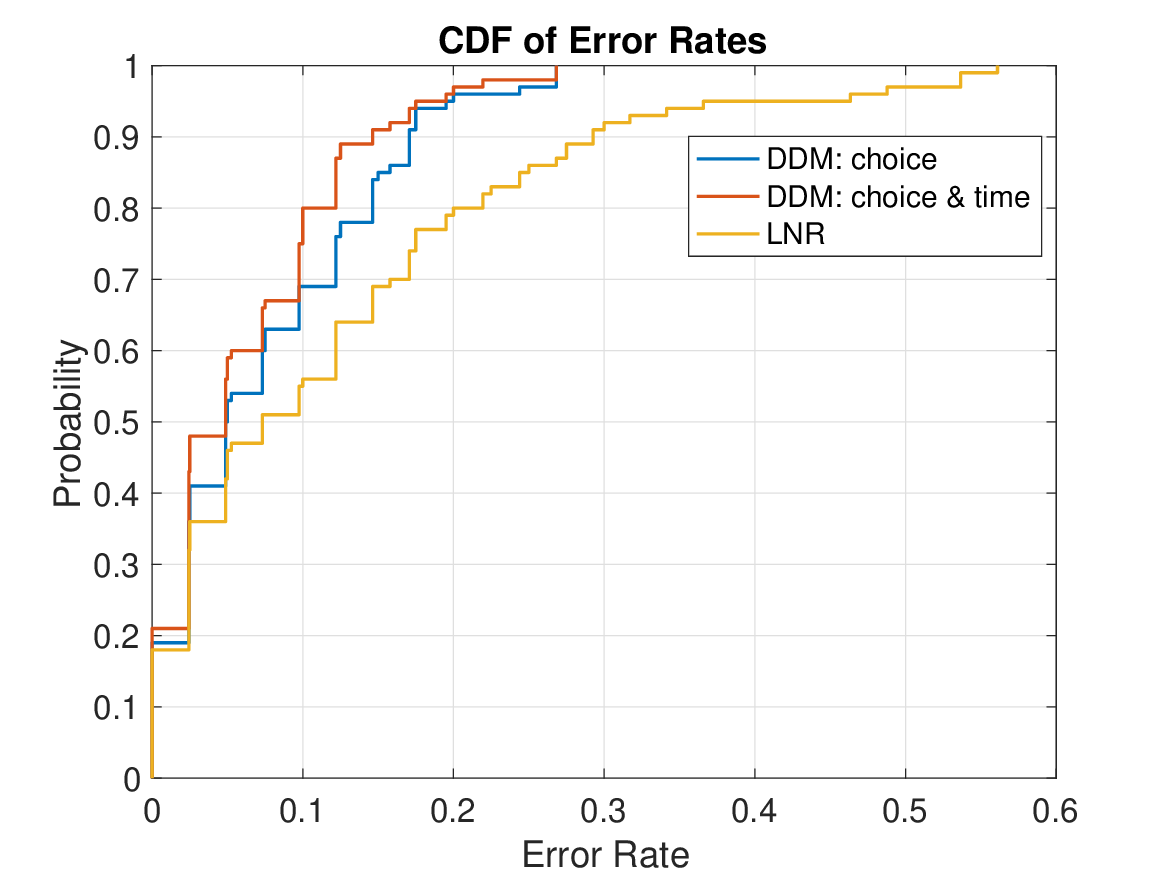}
  \caption{The empirical CDF of error rates}
  \label{fig:prediction-error-CDF} 
\end{subfigure}
\begin{subfigure}{.49\textwidth}
  \centering
  \includegraphics[width=1\linewidth]{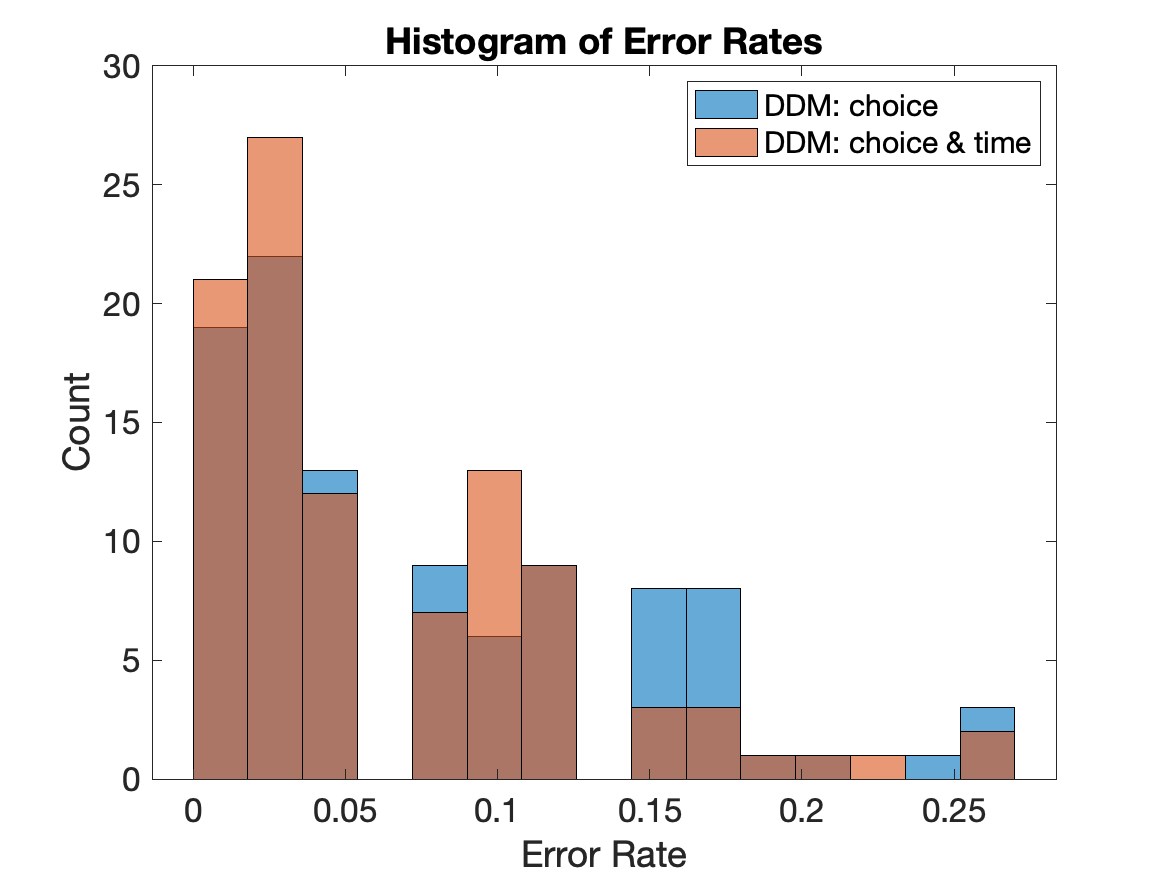}
  \caption{The histogram of error rates}
  \label{fig:prediction-error-Hist}
\end{subfigure}
\caption{The error rate in predicting agent's choice over the test data for DDM and LNR models}
\label{fig:prediction-error}
\end{figure}
First, all three models result in non-trivial error rates, while the two DDM models demonstrate better performance in predicting the agent’s choices compared to the LNR model. Our results show that changing the specification of the LNR does not lead to significant improvements in predictive ability, and the DDM models remain superior. Furthermore, between the two DDM models, the model trained with both time and choice data performs better, as the empirical CDF of the model trained with choice data alone almost exhibits first-order stochastic dominance over that of the model trained with both time and choice data. This indicates that incorporating response-time data leads to a better-trained model and higher predictive accuracy. The LNR model's performance is worse even in the sense of first-order stochastic dominance.

\paragraph{Predicting the response time:}
We next evaluate the model trained using both choice and response-time data to assess its ability to predict response times. We make two remarks in this regard. First, as discussed in \cref{sec:DDM}, when the boundary parameter $b$ is unknown, our formulation instead learns $\bw^*/b$. This does not affect a binary prediction of the agent’s choice, since choice behavior is invariant to the scaling of $\bw^*$. However, an estimate of $b$ is essential for predicting response times. To address this, we apply a moment-matching method to estimate the value of $b$ that best predicts the average response time, given the trained estimate of $\bw^*/b$.

Second, even with a perfectly specified DDM and for a given pair of alternatives, response time remains a continuous random variable. Moreover, since the pairs of alternatives vary across our sample data, the underlying distribution of response times also changes, further complicating the prediction task. To capture this uncertainty, we report a coverage interval for the predicted response time: specifically, we construct an interval centered at the mean response time predicted by the trained DDM, with a radius equal to one standard deviation.

\cref{fig:prediction-RT-samples} illustrates this coverage interval and its mean, together with the observed response times for the first 200 samples. Additionally, for each agent, we compute the miscoverage rate---defined as the fraction of samples for which the observed response time falls outside the predicted coverage interval. The histogram of miscoverage rates across all agents is shown in \cref{fig:prediction-RT-Hist}. The average miscoverage rate is 0.051.

\begin{figure}
\centering
\begin{subfigure}{.49\textwidth}
  \centering
  \includegraphics[width=1\linewidth]{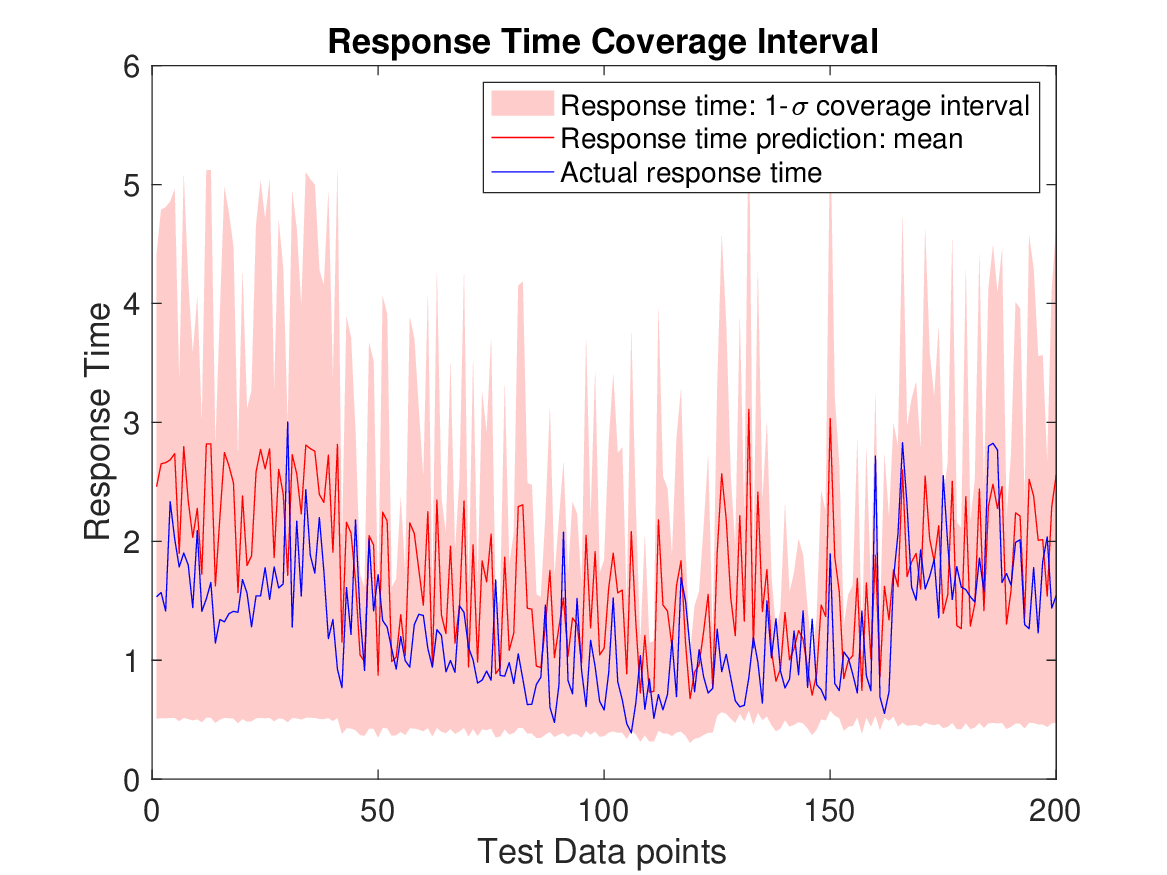}
  \caption{Coverage interval for 200 samples}
  \label{fig:prediction-RT-samples} 
\end{subfigure}
\begin{subfigure}{.49\textwidth}
  \centering
  \includegraphics[width=1\linewidth]{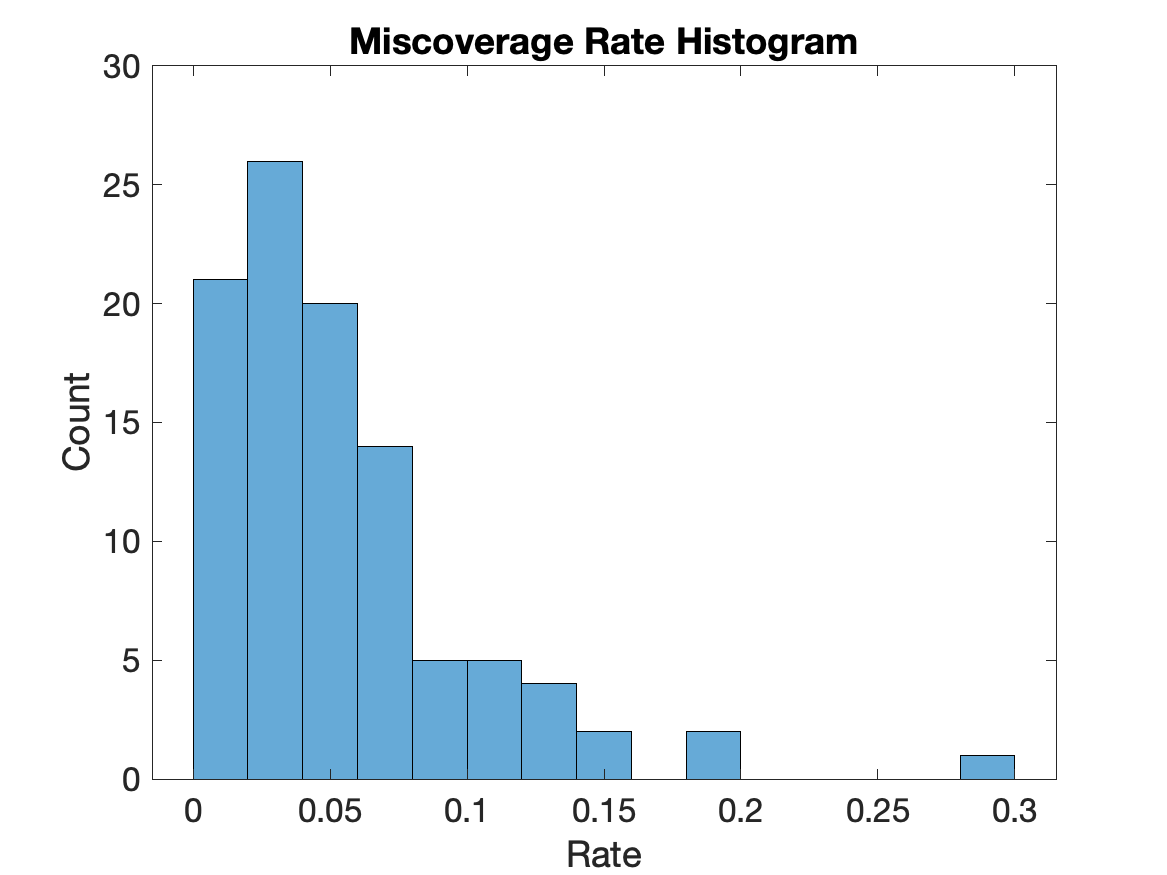}
  \caption{Histogram of miscoverage rate}
  \label{fig:prediction-RT-Hist}
\end{subfigure}
\caption{Predicting the response time}
\label{fig:prediction-RT}
\end{figure}
\paragraph{Interpreting the parameter estimates:} 
Recall that each alternative in the experiment is a dated reward, $[\text{money}, \text{time}]$. Therefore, for the linear DDM models, the two parameters can be interpreted through the effect of the money-delay intertemporal tradeoff on the drift rate (the relative utility of the two alternatives in each choice problem). In \cref{fig:model-weights}, we plot the results for the two trained DDM models side by side. Each plot shows all individual parameter estimates (all the 100 trained models). As one can see, the weight for money is always nonnegative, which is expected. The coefficient on time is generally negative, but there are a few subjects for whom we obtain a positive estimate. 

\begin{figure}
\centering
\begin{subfigure}{.49\textwidth}
  \centering
  \includegraphics[width=1\linewidth]{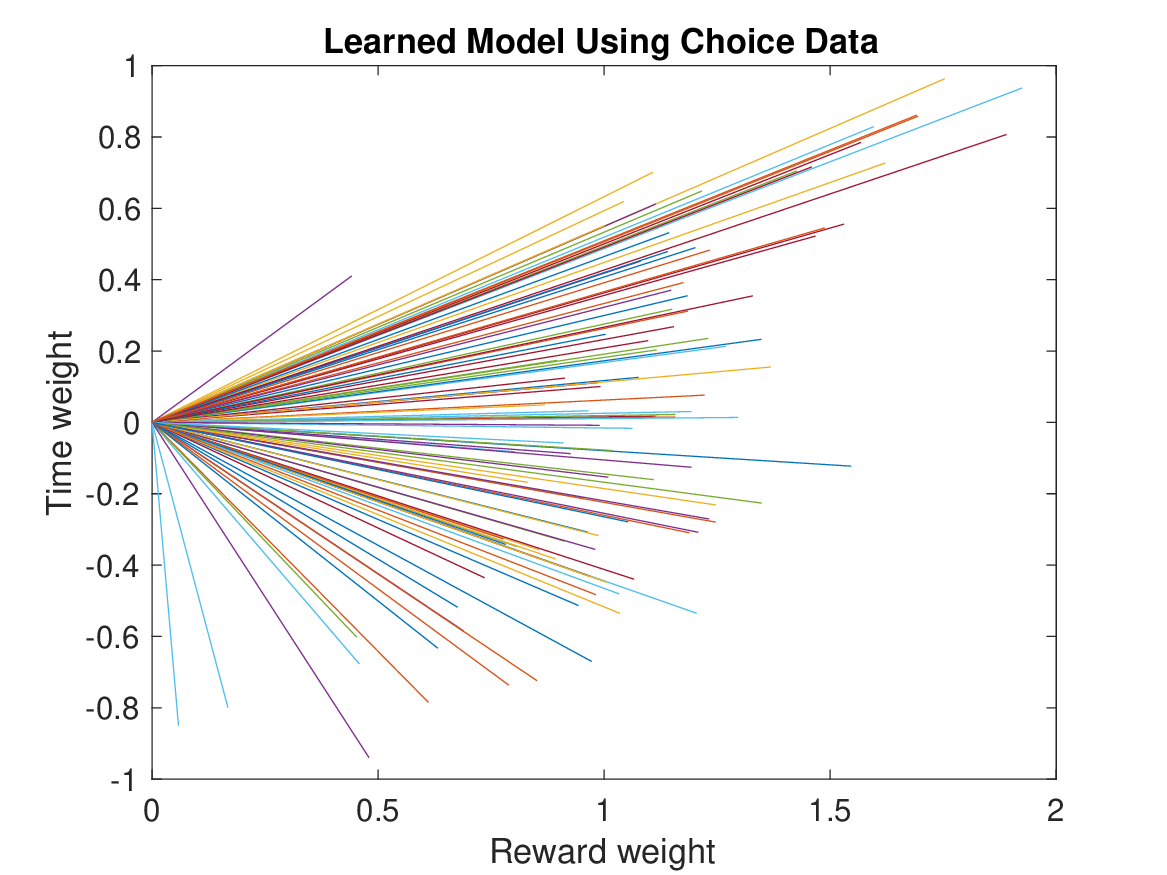}
  \caption{Trained over choice data}
  \label{fig:model-weights-1} 
\end{subfigure}
\begin{subfigure}{.49\textwidth}
  \centering
  \includegraphics[width=1\linewidth]{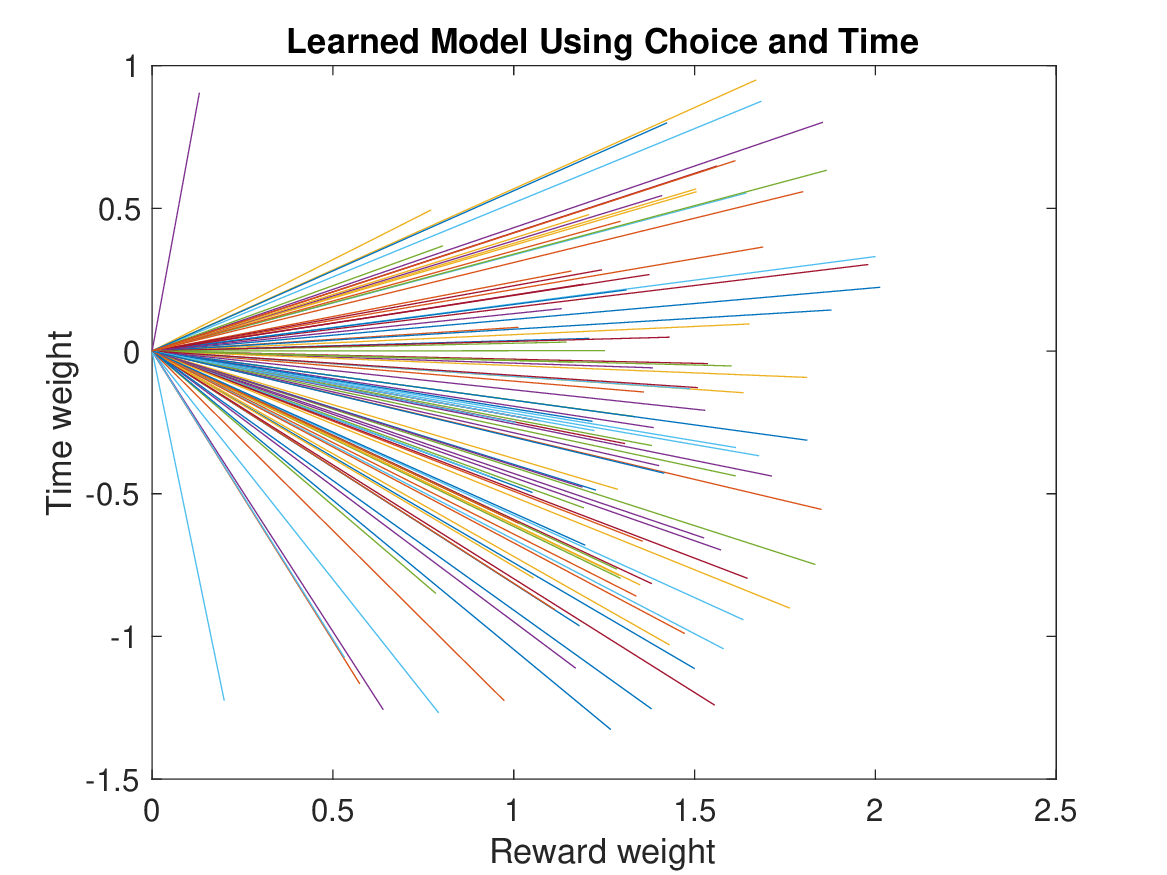}
  \caption{Trained over choice and response-time data}
  \label{fig:model-weights-2}
\end{subfigure}
\caption{Models' weights}
\label{fig:model-weights}
\end{figure}

Given our linear specification, we can compare the parameter estimates and offer an economic interpretation to the estimated values. In particular, if we assume a CARA utility over monetary rewards, we can interpret the implied discount factor in our estimates. Note that the linear utility over dated reward $(R,T)$ is given by 
\begin{equation}
w_r \cdot R - w_t \cdot T,    
\end{equation}
with parameter vector $\bw^* = [w_r, -w_t]^\top$. 
This utility is ordinally equivalent to 
\begin{equation}
\exp(R) \cdot \exp\left(-\frac{w_t}{w_r}\right)^T,    
\end{equation}
where $\exp(R)$ reflects the CARA assumption and the term $\exp\left(-w_t / w_r\right)$ can be interpreted as an exponential discount factor: the smaller this term, the more patient the agent in question. This discount factor provides a way to compare the two trained DDM models (the first trained only on choice data, and the second trained on both choice and response-time data). Specifically, for each agent, we examine the ratio of the discount factor from the second model to that from the first model. The empirical CDF and histogram of this ratio across all 100 agents are shown in \cref{fig:disc-ratio}.

\begin{figure}
\centering
\begin{subfigure}{.49\textwidth}
  \centering
  \includegraphics[width=1\linewidth]{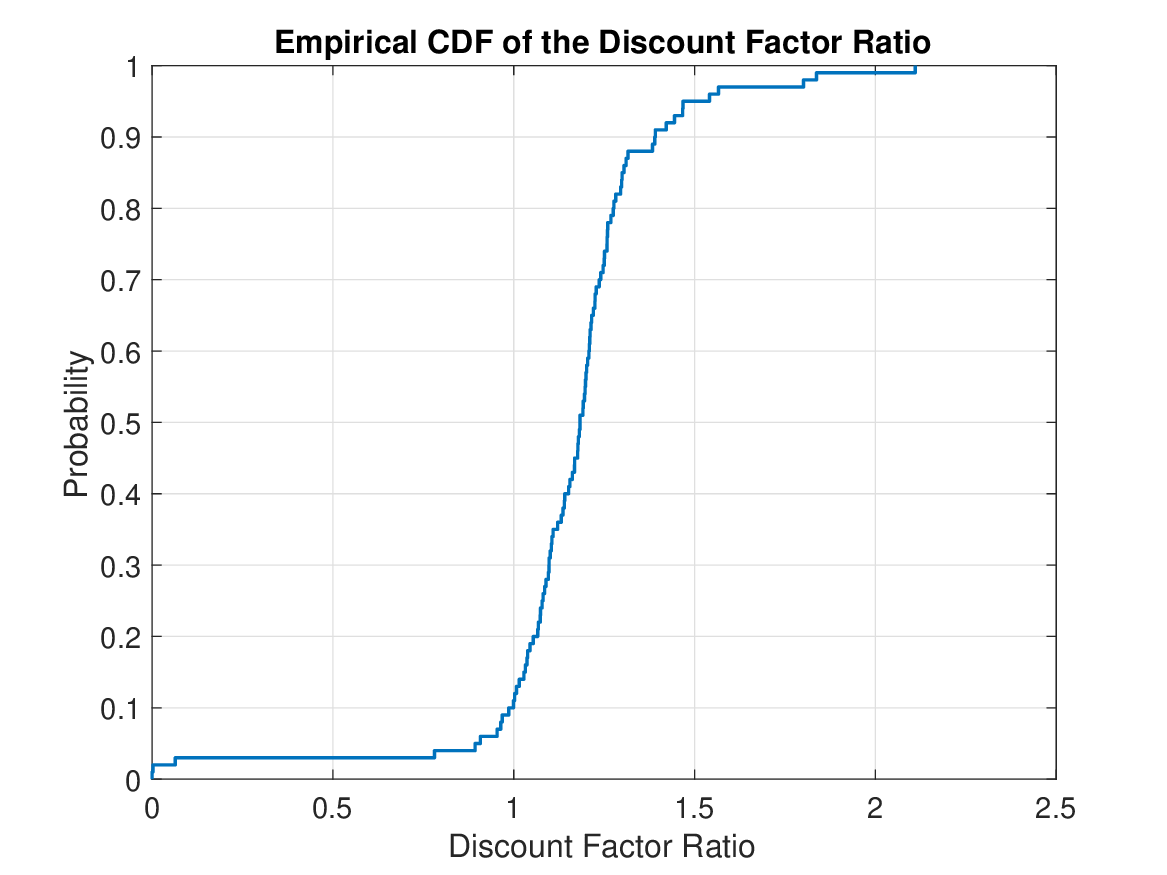}
  \caption{Empirical CDF of the discount factor ratio}
  \label{fig:disc-ratio-CDF} 
\end{subfigure}
\begin{subfigure}{.49\textwidth}
  \centering
  \includegraphics[width=1\linewidth]{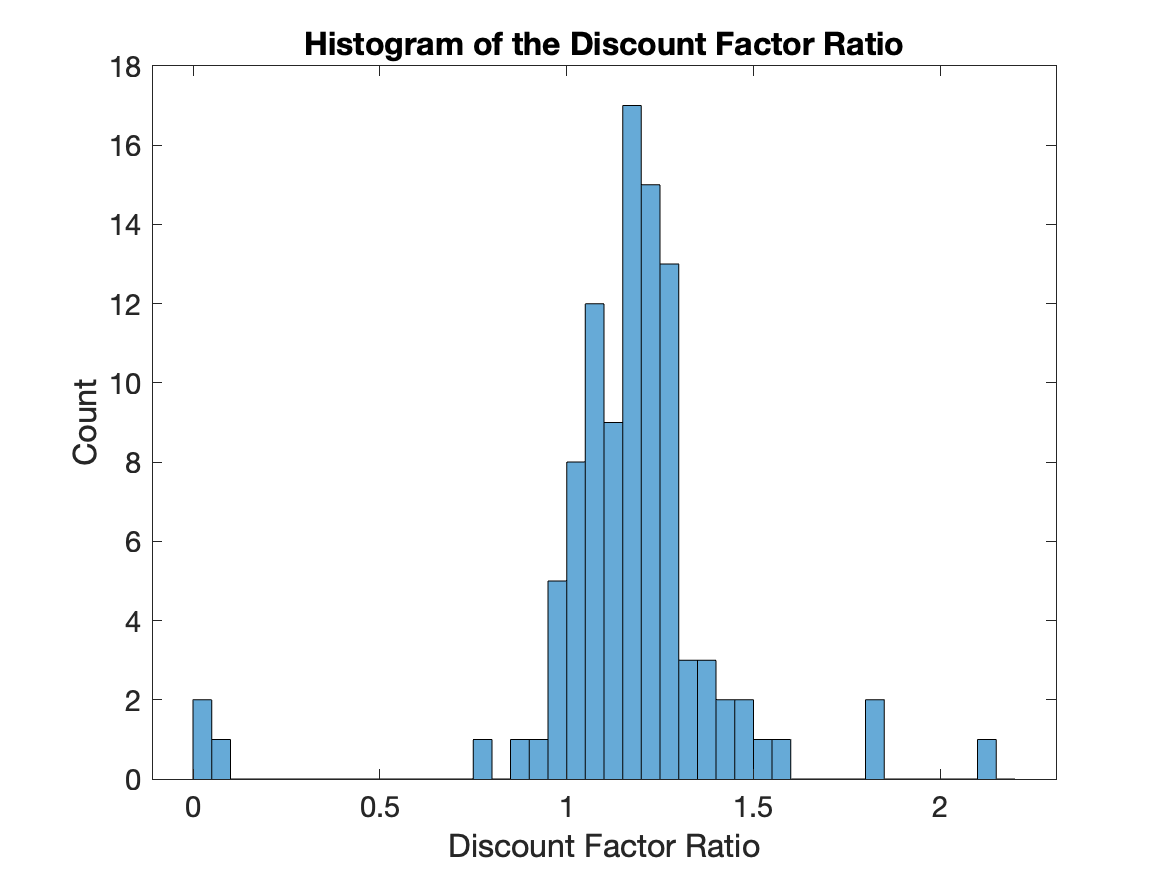}
  \caption{Histogram of the discount factor ratio}
  \label{fig:disc-ratio-hist}
\end{subfigure}
\caption{The discount factor ratio over the two models}
\label{fig:disc-ratio}
\end{figure}
Interestingly, for almost 90\% of agents, the estimated discount factor from the DDM model using response-time data is larger than that from the model using only choice data. The model with choice data also predicts the agents' behavior more accurately, so our findings suggest that the standard methodology in economics that solely relies on choice data leads to a slight underestimation of discount factors.

\section{Proofs}\label{sec:proofs}

\subsection{Proofs From Section~\ref{sec:DDM}}\label{sec:proofsddm}

The following lemma presents the distribution of $z(x,y)$ and $t(x,y)$ under the DDM model. The proof is provided in the appendix.  
\begin{lemma} \label{lemma:BM_distributions}
Under the DDM model~\eqref{eqn:DDM_model_original}, the following holds:
\begin{itemize}
\item The distribution of $z(\bx,\by)$ is given by
\begin{equation}
p(z;\bx,\by) = \frac{1}{1+\exp(-2bzv(\bx,\by;\bw^*))}.
\end{equation}
Moreover, the expectation of $z(\bx,\by)$ is given by
\begin{equation}
\mathbb{E}[z(\bx,\by)] = \tanh(bv(\bx,\by;\bw^*)).  
\end{equation}
\item The distribution of $t(\bx,\by)$, conditioned on $z(\bx,\by) = z$, is given by
\begin{equation}
f(t|z(\bx,\by) = z;\bx,\by) = 
\cosh \left(bv(\bx,\by;\bw^*) \right) \exp\left(-\frac{v(\bx,\by;\bw^*)^2}{2}t\right) \varphi(t),    
\end{equation}
where $\varphi(t)$ denotes the PDF of $T^b$ for the standard Brownian motion (i.e., when the drift is zero), and is given by
\begin{equation}
\varphi(t) = \frac{2b}{\sqrt{2\pi t^3}}\sum_{n=-\infty}^\infty (4n+1) \exp \left(-\frac{(4n+1)^2b^2}{2t} \right).    
\end{equation}
Moreover, the expectation of $t(\bx,\by)$ is given by
\begin{equation}
\mathbb{E}[t(\bx,\by)] = \frac{b}{v(\bx,\by;\bw^*)} \tanh(bv(\bx,\by;\bw^*)).  
\end{equation}
\end{itemize}
\end{lemma}
\subsubsection*{Proof of \cref{theorem:linear_DDM}}
For a random sample $(\bx,\by,z,t)$, define the following interim variables:
\begin{align}
\ba &:= \frac{\sqrt{t}(\bx-\by)}{b}, \\
\bxi &:= \frac{(\bx-\by)z}{b} - \ba^\top \bw^* \ba 
= \frac{1}{b} \left (z - \frac{t}{b} (\bx-\by)^\top \bw^* \right) (\bx-\by),\\
H &:= \mathbb{E}_{\bx,\by,t} \left [ \ba \ba^\top \right]
= \frac{1}{b^2} ~\mathbb{E}\left[ t(\bx-\by) (\bx-\by)^\top \right].
\end{align}
The result of \cite{bach2013non} establishes that, if we could find $\sigma^2$ and $R^2$ such that
\begin{subequations}\label{eqn:bach_conditions}
\begin{align}
& \mathbb{E}\left[ \bxi \bxi^\top \right] \preceq \sigma^2 H, \label{eqn:bach_condition_1} \\
& \mathbb{E}\left[ \|\ba\|^2 \ba \ba^\top \right] \preceq R^2 H, \label{eqn:bach_condition_2}   
\end{align}    
\end{subequations} 
with $A \preceq B$ meaning that $B-A$ is positive semi-definite, then, by setting $\lambda = 1/(4R^2)$, we have
\begin{equation} \label{eqn:error_loss}
\mathbb{E} \left[ \mathcal{L} \left( \frac{1}{n+1}\sum_{i=0}^{n}\hat{\bw}_i \right) - \mathcal{L}(\bw^*)   \right] 
\leq \frac{2}{n+1} \left ( \sigma \sqrt{d} + R \|\bw^* - \hat{\bw}_0 \| \right) ^2.
\end{equation}
To show the conditions of \eqref{eqn:bach_conditions}, we first prove the following lemma on the Laplace transform of the response time. The proof is deferred to the appendix.
\begin{lemma} \label{lemma:laplace_DDM}
For any $\alpha \geq 0$ and pair of alternatives $(\bx,\by)$, we have
\begin{equation} \label{eqn:laplace_transform}
\mathbb{E}[\exp(-\alpha t(\bx,\by))] = 
\frac{\cosh\left (b(\bx-\by)^\top \bw^* \right)}{\cosh\left (b\sqrt{2\alpha+ ((\bx-\by)^\top \bw^*)^2} \right)}.
\end{equation}
\end{lemma}
This lemma allows us to compute the second moment of the response time.
\begin{lemma} \label{lemma:response_time_second_moment}
For any pair of alternatives $(\bx,\by)$, we have
\begin{equation}
\mathbb{E}[t(\bx,\by)^2] = \frac{b\sech^2(bv) \left(-3bv +\sinh(2bv)+bv \cosh(2bv) \right)}{2v^3}, 
\end{equation}
with $v = (\bx-\by)^\top \bw^*$.
\end{lemma}
The proof follows from taking the second derivative of the Laplace transform function \eqref{eqn:laplace_transform} as a function of $\alpha$ and setting $\alpha=0$.\footnote{Note that we are in changing the order of the expectation and the derivative. We defer making this step rigorous at this point; we will instead formalize it for all moments of $t$ in \cref{sec:Rademacher_DDM}} 
We next investigate the conditions of \eqref{eqn:bach_conditions}. 

\paragraph{Condition \eqref{eqn:bach_condition_1}:} We claim this condition holds, in fact, as an equality, with $\sigma = 1/b$. To see this, notice that 
\begin{align}
& \mathbb{E}\left[ \bxi \bxi^\top \right] =  \frac{1}{b^2} \mathbb{E} \left [
\left(z - \frac{t}{b} (\bx-\by)^\top \bw^* \right)^2 (\bx-\by) (\bx-\by)^\top \right ] \\
&\quad = \frac{1}{b^2} \mathbb{E} \left[ 
\left ( 1 - \frac{2tz}{b} (\bx-\by)^\top \bw^* + \frac{t^2}{b^2} \left( (\bx-\by)^\top \bw^* \right)^2 \right ) (\bx-\by) (\bx-\by)^\top
\right ] \\
&\quad = \frac{1}{b^2} \mathbb{E}_{(\bx,\by)\sim \mu} \left[ 
\left ( 1 - \frac{1}{b^2}\left( 2\mathbb{E}[t(\bx,\by)]^2 - \mathbb{E}[t(\bx,\by)^2] \right) \left( (\bx-\by)^\top \bw^* \right)^2\right ) (\bx-\by) (\bx-\by)^\top
\right ], \label{eqn:condition_1_1}
\end{align}
where the last equality uses the fact that, conditional on $(\bx, \by)$, $z(\bx, \by)$ and $t(\bx, \by)$ are independent, and that $\mathbb{E}[z(\bx,\by)] = \mathbb{E}[t(\bx,\by)](\bx-\by)^\top \bw^*/b$. 
On the other hand, we have
\begin{equation}
H 
= \frac{1}{b^2} ~\mathbb{E}_{(\bx,\by)\sim \mu} \left[  \mathbb{E}[t(\bx,\by)] (\bx-\by) (\bx-\by)^\top \right].   
\end{equation}
To show $\mathbb{E}\left[ \bxi \bxi^\top \right] = H/b^2$, it suffices to show that
\begin{equation}
1 - \frac{1}{b^2}\left( 2\mathbb{E}[t(\bx,\by)]^2 - \mathbb{E}[t(\bx,\by)^2] \right) \left( (\bx-\by)^\top \bw^* \right)^2 = \frac{1}{b^2} \mathbb{E}[t(\bx,\by)].    
\end{equation}
To see why this holds, we substitute $\mathbb{E}[t(\bx, \by)]$ and $\mathbb{E}[t(\bx, \by)^2]$ from \cref{lemma:BM_distributions} and \cref{lemma:response_time_second_moment}, respectively. This reduces the problem to establishing the following identity:
\begin{equation} \label{eqn:condition_1_2}
1 + \frac{\sech^2(\beta) \left( -3\beta + \sinh(2\beta) + \beta \cosh(2\beta)\right)}{2\beta} =   
2\tanh^2(\beta) + \frac{\tanh(\beta)}{\beta},
\end{equation}
with $\beta := b(\bx-\by)^\top \bw^*$. One can verify that this identity holds for any $\beta$ which completes the proof of this part.
\paragraph{Condition \eqref{eqn:bach_condition_2}:} We claim this condition holds with $R^2= 2D^2$.
To see this, note that
\begin{align}
\mathbb{E}\left[ \|\ba\|^2 \ba \ba^\top \right] &= 
\frac{1}{b^4} \mathbb{E}_{(\bx,\by)\sim \mu} \left [ \mathbb{E}[t(\bx,\by)^2] \|\bx-\by\|^2 (\bx-\by) (\bx-\by)^\top \right] \\
& \preceq \frac{D^2}{b^4} \mathbb{E}_{(\bx,\by)\sim \mu} \left [ \mathbb{E}[t(\bx,\by)^2] (\bx-\by) (\bx-\by)^\top \right], \label{eqn:condition_2_1}
\end{align}
where the last inequality follows from fact that $\|\bx-\by\| \leq D$. Next, note that, for every $(\bx,\by)$, we have
\begin{align}
\frac{b^2 \mathbb{E}[t(\bx,\by)]}{\mathbb{E}[t(\bx,\by)^2]} &= 
\frac{2\beta^2 \tanh(\beta)}{\sech^2(\beta) \left(-3\beta +\sinh(2\beta)+\beta \cosh(2\beta) \right)}.
\end{align}
One can verify that the right-hand side is lower bounded by $0.6$, which plugging it into \eqref{eqn:condition_2_1} yields
\begin{equation}
\mathbb{E}\left[ \|\ba\|^2 \ba \ba^\top \right] \preceq
2 \frac{D^2}{b^2} \mathbb{E}_{(\bx,\by)\sim \mu} \left [ \mathbb{E}[t(\bx,\by)] \|\bx-\by\|^2 (\bx-\by) (\bx-\by)^\top \right] = 2D^2 H,
\end{equation}
which completes the proof of the second condition.
\paragraph{Bounding the matrix norm:}
Hence, we have established \eqref{eqn:error_loss} with $\sigma^2=1/b^2$ and $R^2=2D^2$. Note that $\mathcal{L}(\cdot)$ is given by
\begin{align}
\mathcal{L}(\bw)&= \mathbb{E}_{\bx,\by,z,t}\left[ \frac{t}{2b^2} \bw^\top (\bx-\by) (\bx-\by)^\top \bw - \frac{z}{b} (\bx-\by)^\top \bw \right ]  \nonumber \\
&= \bw^\top \Sigma_1 \bw - \mathbb{E}_{\bx,\by,z} \left[ \frac{z}{b} (\bx-\by)^\top \right] \bw, \label{eqn:quadratic_loss}
\end{align}
with
\begin{equation}
\Sigma_1 := \frac{1}{2b^2}\mathbb{E}_{\bx,\by,t}\left[t(\bx-\by) (\bx-\by)^\top \right] 
= \frac{1}{2b^2} \mathbb{E}_{\bx,\by}\left[\mathbb{E}[t(\bx,\by)](\bx-\by) (\bx-\by)^\top \right].  
\end{equation}
Next, notice that by the definition of $\bw^*$, we have
\begin{equation}
2\Sigma_1 \bw^* = \mathbb{E}_{\bx,\by,z} \left[ \frac{z}{b} (\bx-\by)\right]. 
\end{equation}
Substituting this into \eqref{eqn:quadratic_loss} implies
\begin{equation}
\mathcal{L}(\bw) = \bw^\top \Sigma_1 \bw - 2 \bw^\top \Sigma_1 \bw^*.    
\end{equation}
In the special case $\bw = \bw^*$, this yields
\begin{equation}
\mathcal{L}(\bw^*) = - {\bw^*}^\top \Sigma_1 \bw^*.  
\end{equation}
As a result, we have
\begin{equation}
\mathcal{L}(\bw) - \mathcal{L}(\bw^*) = \|\bw-\bw^*\|_{\Sigma_1}^2.
\end{equation}
Therefore, we have established
\begin{equation} \label{eqn:error_matrix_norm_1}
\mathbb{E} \left [ \left \| \frac{1}{n+1}\sum_{i=0}^{n}\hat{\bw}_i - \bw^*
\right \|_{\Sigma_1}^2 \right]   \leq
\frac{4}{n+1} \left (\frac{d}{b^2} + 2D^2 \|\hat{\bw}_0 - \bw^*\|^2 \right).
\end{equation}
Note that we are seeking a result in terms of a matrix norm with respect to the matrix $\Sigma$, whereas the current expression involves the matrix $\Sigma_1$. To translate this into a matrix norm with respect to $\Sigma$, notice that
\begin{equation} \label{eqn:lower-bound-Etime}
\mathbb{E}[t(\bx,\by)] = b^2 \frac{\tanh(\beta)}{\beta}
\geq \frac{b^2}{\sqrt{1+\beta^2}},
\end{equation}
where, recall, $\beta = b(\bx - \by)^\top \bw^*$; see also \cite{bagul2022tight} for a lower bound on the hyperbolic tangent function.
This implies
\begin{equation}
\Sigma_1 \succeq \frac{1}{2\sqrt{1+b^2 D^2\|\bw^*\|^2}} \Sigma,
\end{equation}
which allows us to translate \eqref{eqn:error_matrix_norm_1} into a result involving the matrix norm with respect to $\Sigma$, thereby completing the proof of \cref{theorem:linear_DDM}.
\subsubsection*{Proof of \cref{theorem:generalization_DDM}}
Let 
\begin{equation}
\ell(\bx,\by,z,t;\bw) := \frac{t}{2} \left( \frac{v(\bx,\by;\bw)}{b}\right)^2- z \frac{v(\bx,\by;\bw)}{b}.   
\end{equation}
Note that $\mathcal{L}(\cdot)$ is the expectation of $\ell(\bx, \by, z, t; \bw)$. Recall that the Rademacher complexity of the function $\ell$ with respect to $\mathcal{S}$ is defined as as
\begin{equation} 
\mathcal{R}_n \Big( \{\ell(\cdot;\bw)\}_{\bw} \circ \mathcal{S} \Big):= 
\frac{1}{n} \mathbb{E}_{\bomega}\left[ \sup_{\bw \in \mathcal{W}} \sum_{i=1}^n \omega_i \ell(\bx_i,\by_i,z_i,t_i;\bw) \right].
\end{equation}
As stated in \citet[Theorem 26.3]{shalev2014understanding}, we have
\begin{equation}\label{eqn:Rademacher-to-generalization}
\mathbb{E}_{\mathcal{S}} \left [ \sup_{\bw \in \mathcal{W}} \left(\mathcal{L}(\bw) - \hat{\mathcal{L}}(\bw) \right) \right]  \leq 
2 \mathbb{E}_{\mathcal{S}} \left [ \mathcal{R}_n \Big( \{\ell(\cdot;\bw)\}_{\bw} \circ \mathcal{S} \Big) \right].
\end{equation}
Next, note that we have
\begin{align}
\mathcal{R}_n \Big( \{\ell(\cdot;\bw)\}_{\bw} \circ \mathcal{S} \Big)  & \leq
\mathcal{R}_n \left (\left \{ \frac{t}{2} \left( \frac{v(\bx,\by;\bw)}{b}\right)^2 \right\}_{\bw} \circ \mathcal{S} \right)  
+ \mathcal{R}_n \left ( \left \{  z \frac{v(\bx,\by;\bw)}{b}\right\}_{\bw} \circ \mathcal{S} \right) \label{eqn:proof-rademacher-sup} \\
& \leq \frac{\max_{i} t_i}{2b^2}  ~\mathcal{R}_n \Big (\{v(\cdot, \cdot; \bw)^2\}_{\bw} \circ (\bx_i,\by_i)_{i=1}^n \Big)  
+ \frac{1}{b} \mathcal{R}_n \Big (\{v(\cdot, \cdot; \bw)\}_{\bw} \circ (\bx_i,\by_i)_{i=1}^n \Big) \label{eqn:proof-rademacher-contraction} \\ 
& \leq \left (\frac{K}{b^2} \max_{i} t_i + \frac{1}{b} \right) \mathcal{R}_n \Big (\{v(\cdot, \cdot; \bw)\}_{\bw} \circ (\bx_i,\by_i)_{i=1}^n \Big) \label{eqn:proof-rademacher-contraction-2},
\end{align}
where \eqref{eqn:proof-rademacher-sup} follows from upper bounding the supremum of the sum of two functions by the sum of their suprema, and \eqref{eqn:proof-rademacher-contraction} follows from applying the contraction property of Rademacher complexity \citep[Lemma 26.9]{shalev2014understanding} to the functions $x \mapsto t_i x$ and $x \mapsto z_i x$ (applied on the $i$-th entry). Finally, \eqref{eqn:proof-rademacher-contraction-2} also follows from the contraction property, this time applied to the function $x \mapsto x^2$, which is $2K$-Lipschitz on the interval $[-K, K]$. 

We can now substitute \eqref{eqn:proof-rademacher-contraction-2} into \eqref{eqn:Rademacher-to-generalization}, and it remains to bound the expectation of $\max_i t_i$, conditioned on $(\bx_i,\by_i)_{i=1}^n$. To do so, note that, for any $q \geq 1$, we have
\begin{align}
\left (\mathbb{E} \left [\max_{i} t_i \Big | (\bx_i,\by_i)_{i=1}^n \right] \right)^q & \leq     
\mathbb{E} \left[ (\max_{i} t_i)^q \Big | (\bx_i,\by_i)_{i=1}^n \right] \label{eqn:max_time_1} \\
&\leq \mathbb{E} \left[ \sum_{i=1}^n t_i^q \Big | (\bx_i,\by_i)_{i=1}^n \right]  = \sum_{i=1}^n \mathbb{E}[t(\bx_i,\by_i)^q], \label{eqn:max_time_2}
\end{align}
where \eqref{eqn:max_time_1} follows from Jensen's inequality. As a result, we have
\begin{equation} \label{eqn:max-time-3}
\mathbb{E} \left [\max_{i} t_i \Big | (\bx_i,\by_i)_{i=1}^n \right] \leq \left( \sum_{i=1}^n \mathbb{E}[t(\bx_i,\by_i)^q] \right)^{1/q}.    
\end{equation}
We now need to derive upper bounds for the moments of $t(\bx, \by)$. Let $t_0$ denote the random variable corresponding to the driftless DDM, which is essentially the first time a standard Brownian motion hits either $b$ or $-b$. We make the following claim:
\begin{claim} \label{claim:max-time-driftless}
For any pair of alternatives $(\bx,\by)$ and any $q \geq 1$, we have 
$\mathbb{E}[t(\bx,\by)^q] \leq \mathbb{E}[t_0^q]$.
\end{claim}
\begin{proof}[Proof of Claim \ref{claim:max-time-driftless}]
Note that the distribution of $t_0$ is $\varphi(\cdot)$. Hence, by \cref{lemma:BM_distributions}, the ratio of the PDF of $t_0$ over the PDF of $t(\bx, \by)$ is given by
\begin{equation}
\frac{\exp\left(v(\bx,\by;\bw^*)^2~t/2\right)}{\cosh \left(bv(\bx,\by;\bw^*) \right) },    
\end{equation}
which is increasing in $t$. This implies that the monotone likelihood ratio property holds, which in turn implies that $t_0$ first-order stochastically dominates $t(\bx, \by)$. This immediately establishes the claim.
\end{proof}
Using this claim along with \eqref{eqn:max-time-3} yields
\begin{equation} \label{eqn:max-time-4}
\mathbb{E} \left [\max_{i} t_i \Big | (\bx_i,\by_i)_{i=1}^n \right] \leq n^{1/q}~\mathbb{E}[t_0^q]^{1/q}.    
\end{equation}
Next, we focus on bounding the moments of $t_0$. Recall from \cref{lemma:laplace_DDM} that the Laplace transform of $t_0$ is given by
\begin{equation}
\mathbb{E} \left[ \exp(- \alpha t_0) \right] =\frac{1}{\cosh(b\sqrt{2\alpha})}.   
\end{equation}
We now make the following claim.
\begin{claim} \label{claim:moments-laplace-transform}
For any positive integer $q$, we have    
\begin{equation}
\mathbb{E}[t_0^q] = (-1)^q \frac{d^q}{d\alpha^q} \frac{1}{\cosh(b\sqrt{2\alpha})} \bigg|_{\alpha=0}.   
\end{equation}
\end{claim}
\begin{proof}[Proof of Claim \ref{claim:moments-laplace-transform}]
This essentially states that we can interchange the order of differentiation and expectation. For any $\alpha_0 > 0$ and positive $q$, since $t_0^q \exp(-\alpha t_0)$ is uniformly bounded in a neighborhood around $\alpha_0$, this interchange is allowed by the Leibniz integral rule (or, basically, the dominated convergence theorem). In other words, for any $\alpha_0 > 0$, we have
\begin{equation}
\mathbb{E} \left[ t_0^q \exp(-\alpha_0 t^q) \right] = (-1)^q \frac{d^q}{d\alpha^q} \frac{1}{\cosh(b\sqrt{2\alpha})} \bigg|_{\alpha=\alpha_0}.    
\end{equation}
That this holds at $\alpha_0 = 0$ (which is the case relevant to our claim) follows from the monotone convergence theorem.
\end{proof}
Therefore, to compute the moments of $t_0$, we consider the Taylor series expansion of $1/\cosh(b\sqrt{2\alpha})$, which is given by
\begin{equation}
\frac{1}{\cosh(b\sqrt{2\alpha})} = \sum_{k=0}^\infty \frac{2^k~ b^{2k}~E_{2k}}{(2k)!} \alpha^k,     
\end{equation}
for $\alpha \in [0, \pi^2/(8b^2))$, where $E_{2k}$ denotes the $2k$-th Euler number \citep[Chapter 23]{abramowitz1965handbook}. Thus, we have
\begin{equation} \label{eqn:moment-bound-1}
\mathbb{E}[t_0^q] = (-1)^q \frac{2^q~b^{2q}~q!~E_{2q}}{(2q)!}    
\leq \frac{2^{2q+3}~b^{2q}~ q!}{\pi^{2q+1}} 
\leq \frac{2^{2q+3}~b^{2q}~ \sqrt{2\pi~q}q^q}{e^q~\pi^{2q+1}}, 
\end{equation}
where the first inequality follows from an upper bound on the Euler numbers \citep[Chapter 23]{abramowitz1965handbook} and the second inequality follows from Stirling's approximation. Substituting \eqref{eqn:moment-bound-1} into \eqref{eqn:max-time-4} implies 
\begin{equation} \label{eqn:max-time-5}
\mathbb{E} \left [\max_{i} t_i \Big | (\bx_i,\by_i)_{i=1}^n \right] \leq n^{1/q}~\frac{8b^2q}{e\pi^2} (8\sqrt{2q}/\sqrt{\pi})^{1/q}.    
\end{equation}
Lastly, setting $q = \log(n)$ and substituting the resulting bound into \eqref{eqn:proof-rademacher-contraction-2}, and then plugging the entire expression into \eqref{eqn:Rademacher-to-generalization}, completes the proof of \eqref{eqn:DDM_generalization_bound}. The proof of \eqref{eqn:DDM_generalization_bound_2} follows directly from the standard decomposition of the total error into training and test errors~\citep[see, e.g.,][Chapter 26]{shalev2014understanding}.
\subsubsection*{Proof of \cref{proposition:e-DDM-ratio}}
Define $\nu(\zeta)$ as the probability that the agent chooses $\bx$ over $\by$, conditioned on starting at $\zeta$, i.e.,
\begin{equation}
\nu(\zeta):= p\left (z(\bx,\by) = 1|W_0=\zeta;\bx,\by \right).    
\end{equation}
 We now use the following identity~\citep[][Theorem 4.2]{taylor_karlin_stochastic}:
\begin{equation}
\mathbb{E} \left [ t(\bx,\by) |W_0 = \zeta \right] = \frac{1}{v(\bx,\by;\bw^*)}  \left( b (2 \nu(\zeta)-1) - \zeta \right).    
\end{equation}
Note that $2\nu(\zeta) - 1$ is, in fact, equal to $\mathbb{E} \left[ z(\bx, \by) \mid W_0 = \zeta \right]$. This, together with the fact that $\pi_0(\cdot)$ has zero mean, completes the proof. 

\subsection{Proof From Section~\ref{sec:LNR}} \label{sec:proofsrace}
\subsubsection*{Proof of \cref{proposition:LNR-ratio}}
To simplify the notation throughout the proof, we define $\nux:=D_0-\nu(\bx;\bw^*)$ and $\nuy:=D_0-\nu(\by;\bw^*)$. We start by noticing that
\begin{equation}\label{eqn:time-LNR-dist}
(\taux, \tauy) \sim \text{Lognormal} \left( 
\begin{bmatrix} \nux \\ \nuy \end{bmatrix},
\begin{bmatrix}
2 & \rho \\
\rho & 2
\end{bmatrix}
\right).
\end{equation}
As a consequence, the distribution of $\taux-\tauy$ is Lognormal$\left( \nux - \nuy, 4-2\rho \right)$. This implies
\begin{equation}
p(z=-1;\bx,\by) = \mathbb{P}(\taux-\tauy > 0) = 1-\Phi \left( \frac{\nuy - \nux}{\sqrt{4-2\rho}} \right),    
\end{equation}
which establishes \eqref{eqn:expected-z-LNR}. Next, we show \eqref{eqn:expected-t-LNR}. Let us define $\htaux:=\log(\taux)$ and $\htauy:=\log(\tauy)$. Notice that the PDF of $(\htaux, \htauy)$ is given by
\begin{equation}
\frac{1}{2\pi\sqrt{4-\rho^2}} \exp \left( 
-\frac{(\htaux-\nux)^2 + (\htauy-\nuy)^2}{4-\rho^2} + \frac{\rho}{4-\rho^2}(\htaux-\nux)(\htauy-\nuy)
\right).    
\end{equation}
Therefore, we have
\begin{subequations} \label{eqn:ex-t-integrals-LNR}
\begin{align} 
\mathbb{E}&[t(\bx,\by)] \nonumber \\
& = \frac{1}{2\pi\sqrt{4-\rho^2}} \int_{-\infty}^\infty \int_{-\infty}^{\htaux} \exp \left(\htauy 
-\frac{(\htaux-\nux)^2 + (\htauy-\nuy)^2}{4-\rho^2} + \frac{\rho}{4-\rho^2}(\htaux-\nux)(\htauy-\nuy)
\right) ~d\htauy d\htaux  \label{eqn:ex-t-integrals-LNR-1} \\
& + \frac{1}{2\pi\sqrt{4-\rho^2}} \int_{-\infty}^\infty \int_{\htaux}^\infty \exp \left(\htaux 
-\frac{(\htaux-\nux)^2 + (\htauy-\nuy)^2}{4-\rho^2} + \frac{\rho}{4-\rho^2}(\htaux-\nux)(\htauy-\nuy)
\right)~d\htauy d\htaux.\label{eqn:ex-t-integrals-LNR-2} 
\end{align}
\end{subequations}
We calculate the first integral \eqref{eqn:ex-t-integrals-LNR-1}, and the second one \eqref{eqn:ex-t-integrals-LNR-2} can be computed similarly. Note that the term inside the exponential in \eqref{eqn:ex-t-integrals-LNR-1} can be recast as
{\small
\begin{align} \label{eqn:recast-exponent}
& \htauy -\frac{(\htaux-\nux)^2 + (\htauy-\nuy)^2}{4-\rho^2} + \frac{\rho}{4-\rho^2}(\htaux-\nux)(\htauy-\nuy) \nonumber \\
&= -\frac{1}{4-\rho^2} \left( \htauy-\nuy-\frac{4-\rho^2}{2} - \frac{\rho}{2}(\htaux-\nux) \right)^2 - \frac{(\htaux-\nux)^2}{4-\rho^2}+\frac{4-\rho^2}{4} + \frac{\rho^2}{4(4-\rho^2)}(\htaux-\nux)^2 + \nuy + \frac{\rho(\htaux-\nux)}{2} \nonumber \\
& = -\frac{1}{4-\rho^2} \left( \htauy-\nuy-\frac{4-\rho^2}{2} - \frac{\rho}{2}(\htaux-\nux) \right)^2 - \frac{(\htaux-\nux)^2}{4}+\frac{4-\rho^2}{4} + \nuy + \frac{\rho(\htaux-\nux)}{2}
\end{align}}
Also, note that we have
\begin{align} \label{eqn:squared-term}
& \frac{1}{\sqrt{\pi(4-\rho^2)}} \int_{-\infty}^{\htaux} \exp \left( -\frac{1}{4-\rho^2} \left( \htauy-\nuy-\frac{4-\rho^2}{2} - \frac{\rho}{2}(\htaux-\nux) \right)^2 \right)~d\htauy \nonumber  \\
& = \Phi \left( \frac{\sqrt{2} \left( \htaux(1-\rho/2) - \nuy - \frac{4-\rho^2}{2} +\nux \rho/2 \right)}{\sqrt{4-\rho^2}} \right).
\end{align}
Now, using \eqref{eqn:recast-exponent} along with \eqref{eqn:squared-term}, we derive that the integral in \eqref{eqn:ex-t-integrals-LNR-1} is equal to
{\small
\begin{align}
\frac{1}{2\sqrt{\pi}} \int_{-\infty}^\infty \exp \left( - \frac{(\htaux-\nux)^2}{4} + \frac{\rho(\htaux-\nux)}{2} +\frac{4-\rho^2}{4} + \nuy \right)  \Phi \left( \frac{ \htaux(1-\rho/2) - \nuy - \frac{4-\rho^2}{2} +\nux \rho/2}{\sqrt{(4-\rho^2)/2}} \right) d\htaux   
\end{align}}
which is equal to
\begin{align}
& \frac{1}{2\sqrt{\pi}} \int_{-\infty}^\infty \exp \left( - \frac{(\htaux-\nux - \rho)^2}{4} + \nuy +1 \right)  \Phi \left( \frac{ \htaux(1-\rho/2) - \nuy - \frac{4-\rho^2}{2} +\nux \rho/2}{\sqrt{(4-\rho^2)/2}} \right) d\htaux   \\
& = \exp(\nuy+1) \mathbb{E}_{\htaux \sim \mathcal{N}(\nux+\rho,2)} \left [ \Phi \left( \frac{ \htaux(1-\rho/2) - \nuy - \frac{4-\rho^2}{2} +\nux \rho/2}{\sqrt{(4-\rho^2)/2}} \right) \right] \\
& = \exp(\nuy+1) \Phi \left( \frac{\nux-\nuy+\rho-2}{\sqrt{4-2\rho}} \right),
\end{align}
where the last equality follows from the well-known identity for the expectation of the CDF of a normal distribution \citep{owen1980table}. Substituting this into \eqref{eqn:ex-t-integrals-LNR-1} and computing \eqref{eqn:ex-t-integrals-LNR-2} similarly completes the proof.

\subsection{Proofs From Section~\ref{sec:perceptron}} \label{sec:proofhighprob}
\subsubsection*{Proof of \cref{proposition:translate-to-high-probability}}
First note that, by \cref{theorem:linear_DDM}, we can find $\hat{\bw}$ such that
\begin{equation} \label{eqn:high-prob-proof-0}
\mathbb{E} \left[ \left \|\hat{\bw} - \bw^* \right \|_{\Sigma}^2 \right] \leq C_e:= \frac{8}{n+1} \sqrt{1+b^2} \left (\frac{d}{b^2} + 2 \right),   
\end{equation}
where $C_e$ is, in fact, order of $\mathcal{O}(d/(n\min\{b,b^2\}))$.

Next, note that
\begin{align}
& \mu \left( \left\{ (\bx,\by) : \text{sign}((\bx - \by)^\top \hat{\bw}) \neq  \text{sign}((\bx - \by)^\top \bw^*) \right \} \right) \nonumber \\
&  \leq
\mu \left( \left\{ (\bx,\by) : \Big | (\bx - \by)^\top (\hat{\bw} - \bw^*) \Big | \geq \gamma \right \} \right)
+ \mu \left( \left\{(\bx,\by): |(\bx-\by)^\top \bw^*| < \gamma \right\}\right) \label{eqn:high-prob-proof-1} \\
& \leq 
\frac{\left \|\hat{\bw} - \bw^* \right \|_{\Sigma}^2}{\gamma^2}
+ \mu \left( \left\{(\bx,\by): |(\bx-\by)^\top \bw^*| < \gamma \right\}\right),
\end{align}
where the last inequality follows from  Chebyshev's inequality.
Hence, to obtain the desired result, we need to have
\begin{equation} \label{eqn:high-prob-proof-2}
\left \|\hat{\bw} - \bw^* \right \|_{\Sigma}^2 \leq \gamma^2 \varepsilon.     
\end{equation}
Now, notice that, using Markov's inequality, the probability of \eqref{eqn:high-prob-proof-2} not holding is bounded by 
\begin{equation}
\frac{\mathbb{E} \left[ \left \|\hat{\bw} - \bw^* \right \|_{\Sigma}^2 \right]}{\gamma^2 \varepsilon} \leq \frac{C_e}{\gamma^2 \varepsilon}.    
\end{equation}
Therefore, substituting \eqref{eqn:high-prob-proof-0}, we establish that by having
\begin{equation} \label{eqn:n-for-constant-prob}
n \geq \frac{80}{\gamma^2 \varepsilon} \sqrt{1+b^2} \left (\frac{d}{b^2} + 2 \right),  
\end{equation}
the model $\hat{\bw}$ would be an $(\tilde{\varepsilon}, 0.1)$-estimator. 

To improve this to an $(\tilde{\varepsilon}, \delta)$-estimator, suppose we run the algorithm in \cref{theorem:linear_DDM} to achieve an error $\varepsilon/4$ instead, and then repeat this $k$ times (over $k$ independent batches of data, for some $k$ to be chosen later). In total, this would require a sample size that is $4k$ times \eqref{eqn:n-for-constant-prob}, resulting in independent models $\hat{\bw}_1, \dots, \hat{\bw}_k$. Next, for any pair of alternatives $(\bx, \by)$, we output the majority sign among $\{(\bx - \by)^\top \hat{\bw}_j\}_{j=1}^k$, denoted by $h_m(\bx,\by)$. 
Similar to \eqref{eqn:high-prob-proof-1}, we can write
\begin{align}
& \mu \left( \left\{ (\bx,\by) : h_m(\bx,\by) \neq  \text{sign}((\bx - \by)^\top \bw^*) \right \} \right) \nonumber \\
&  \leq
\mu \left( \left\{ (\bx,\by) : \sum_{j=1}^k \mathbbm{1} \left( \left | (\bx - \by)^\top (\hat{\bw}_j - \bw^*) \right | \geq \gamma  \right) \geq \frac{k}{2}
\right \} \right)
+ \mu \left( \left\{(\bx,\by): |(\bx-\by)^\top \bw^*| < \gamma \right\}\right). \label{eqn:high-prob-proof-3}
\end{align}
Hence, our goal is to bound the first term on the right-hand side. Next, note that what we showed earlier essentially implies that, for each $j$, with probability $0.9$, we have
\begin{equation} \label{eqn:high-prob-proof-4}
\mu \left( E_j \right) \leq \frac{\varepsilon}{4}, \text{ with } 
E_j:=\left\{ (\bx,\by) : \Big | (\bx - \by)^\top (\hat{\bw}_j - \bw^*) \Big | \geq \gamma \right \}.
\end{equation}
We call an estimator \textit{good} if condition \eqref{eqn:high-prob-proof-4} holds for it. Each $\hat{\bw}_j$ is a good estimator with probability $0.9$. By Hoeffding’s inequality, the probability that at least one fourth of the estimators are not good is bounded by $\exp(-2k(0.15)^2)$. Therefore, by setting $k = 23 \log(1/\delta)$, we have that with probability at least $1 - \delta$, at least three quarters of the estimators are good. Now, condition on such an event, we have
\begin{align}
\sum_{j=1}^k \mathbbm{1} \left( \left | (\bx - \by)^\top (\hat{\bw}_j - \bw^*) \right | \geq \gamma  \right) &\leq 
\# \text{ of bad estimators } + \sum_{j: \hat{\bw}_j \text{ is good}} \mathbbm{1} \left( (\bx,\by) \in E_j  \right) \\
&\leq \frac{k}{4} + \sum_{j: \hat{\bw}_j \text{ is good}} \mathbbm{1} \left( (\bx,\by) \in E_j  \right).
\end{align}
This implies 
\begin{align}
& \mu \left( \left\{ (\bx,\by) : \sum_{j=1}^k \mathbbm{1} \left( \left | (\bx - \by)^\top (\hat{\bw}_j - \bw^*) \right | \geq \gamma  \right) \geq \frac{k}{2}
\right \} \right) \nonumber \\
& \quad  \leq
\mu \left( \left\{ (\bx,\by) : \sum_{j: \hat{\bw}_j \text{ is good}} \mathbbm{1} \left( (\bx,\by) \in E_j  \right) \geq \frac{k}{4}
\right \} \right) \nonumber \\
& \quad \leq \frac{\sum_{j: \hat{\bw}_j \text{ is good}} \mu(E_j)}{k/4} \label{eqn:high-prob-proof-5} \\
&\quad \leq \varepsilon,
\end{align}
where \eqref{eqn:high-prob-proof-5} follows from Markov's inequality. This completes the proof. 

\bibliography{references}

\newpage
\appendix
\section{Omitted Proofs}
\subsection{Proof of \cref{lemma:generalization-DDM}}
Recall from \eqref{eqn:dif-L} in the proof of \cref{proposition:minimizer} that
\begin{equation} \label{eqn:corollary-DDM-proof-1}
\mathcal{L}(\bw) - \mathcal{L}(\bw^*) = \frac{1}{2} ~ \mathbb{E}_{(\bx,\by)\sim \mu} \left [
\mathbb{E}[t(\bx,\by)] \left (\frac{v(\bx,\by;\bw)}{b} - \frac{v(\bx,\by;\bw^*)}{b} \right)^2
\right].
\end{equation}
Next, note that, similar to the derivation of \eqref{eqn:lower-bound-Etime} in the proof of \cref{theorem:linear_DDM}, we can establish
\begin{equation} 
\mathbb{E}[t(\bx,\by)] \geq \frac{b^2}{\sqrt{1+b^2K^2}}.
\end{equation}
Using this inequality along with the assumption yield
\begin{equation} \label{eqn:corollary-DDM-proof-2}
\mathcal{L}(\bw) - \mathcal{L}(\bw^*) \geq \frac{1}{2\sqrt{1+b^2K^2}} ~ \|\bw-\bw^*\|_{*}^2.
\end{equation}
\subsection{Proof of \cref{lemma:E-DDM-joint-pdf}}
Let $\varphi(\cdot; v)$ denote the probability density function of $T^b$. Note that we have:
\begin{align*}
f(t, z(\bx,\by)=1;\bx,\by) &= 
\varphi(T^b=t, T_{+}^b < T_{-}^b;v) \\
& = \int_{-b}^b \pi_0(\zeta) \varphi(T^b=t, T_{+}^b < T_{-}^b ~|~ W_0 = \zeta;v) ~d \zeta \\
& = \int_{-b}^b \pi_0(\zeta) \varphi(T^{b-\zeta}=t, T_{+}^{b-\zeta} < T_{-}^{b+\zeta} ~|~ W_0 = 0;v) ~d \zeta \\
&= \int_{-b}^b \pi_0(\zeta) \exp\left ( (b-\zeta)v - v^2/t \right) \varphi(T^{b-\zeta}=t, T_{+}^{b-\zeta} < T_{-}^{b+\zeta} ~|~ W_0 = 0;0) ~d \zeta,
\end{align*}
where the last equation follows from Girsanov’s theorem. Finally, we substitute the expression for $\varphi(T^{b - \zeta} = t,, T_{+}^{b - \zeta} < T_{-}^{b + \zeta} \mid W_0 = 0;, 0)$, which corresponds to the driftless standard Brownian motion starting at zero. This result is provided in Section 2.8 of \cite{karatzas1991brownian}, completing the proof.
\subsection{Proof of \cref{lemma:BM_distributions}}
For the proof of the first part, as well as both expectation results, see Chapter VIII.4 in \cite{taylor_karlin_stochastic}. We show the second part for the case of $z=1$ as the other case follows similarly. Note that we have
\begin{align}
f(t|z(\bx,\by) = 1;\bx,\by) &= \frac{f(t,z(\bx,\by) = 1;\bx,\by)}{p(1;\bx,\by)} \nonumber \\
& = \left ( 1+\exp(-2bv(\bx,\by;\bw^*)) \right) f(t,z(\bx,\by) = 1;\bx,\by). \label{eqn:BM_drift_1} 
\end{align}
Next, using Girsanov's theorem, we have
\begin{equation} \label{eqn:BM_drift_2}
f(t,z(\bx,\by) = 1;\bx,\by) = \exp \left (bv(\bx,\by;\bw^*) - \frac{v(\bx,\by;\bw^*)^2}{2}t  \right) \varphi(T^b=t|T_+^b < T_{-}^b),
\end{equation}
where $\varphi(\cdot)$ is the PDF of $T^b$ for the case of standard Brownian motion, i.e., $v(\bx,\by;\bw^*)=0$.

Finally, note that, for the standard Brownian motion, due to its symmetry, we have
\begin{equation} \label{eqn:BM_drift_3}
\varphi(T^b=t|T_+^b < T_{-}^b) = \frac{1}{2} \varphi(t).    
\end{equation}
Plugging \eqref{eqn:BM_drift_3} into \eqref{eqn:BM_drift_2}, and then substituting the result into \eqref{eqn:BM_drift_1}, completes the proof of the second part. The derivation of $\varphi(t)$ is provided in Section 2.8 of \cite{karatzas1991brownian}.
\subsection{Proof of \cref{lemma:laplace_DDM}}
To simplify the notation, we define $v := (\bx-\by)^\top \bw^*$ as the drift parameter for the DDM model. Note that, by \cref{lemma:BM_distributions}, we have
\begin{align}
\mathbb{E}[\exp(-\alpha T^b) ~|~ T_+^b < T_{-}^b]   
&= \int_{0}^\infty \exp(-\alpha t) \cosh(bv) \exp \left(-\frac{v^2}{2}t \right) \varphi(t) dt \\
& = \cosh(bv) \int_{0}^\infty \exp \left (-\left(\alpha+\frac{v^2}{2} \right) t \right ) \varphi(t) dt, \label{eqn:laplace_transform_1}
\end{align}
where the last integral can be interpreted as the Laplace transform of standard Brownian motion with zero drift, and is equal to
\begin{equation}
\frac{1}{\cosh \left(b\sqrt{2\alpha + v^2}\right)}.    
\end{equation}
Plugging this into \eqref{eqn:laplace_transform_1} completes the proof of this lemma.

\end{document}